\documentclass{llncs}
\usepackage[utf8]{inputenc}

\usepackage{qip}
\usepackage[adversary,asymptotics,primitives,complexity,mm,ff,keys]{cryptocode}

\usepackage{enumitem}

\usepackage{xcolor}
\usepackage{graphics}
\usepackage[colorlinks]{hyperref}

\usepackage{cleveref}

\usepackage[minalphanames=3,maxalphanames=4,maxnames=99,style=alphabetic,url=false]{biblatex}
\newtheorem{result}{Result}

\spnewtheorem{claim}{Claim}{\itshape}{}

\newcommand{\qpt}{\ensuremath{\mathsf{QPT}}}
\newcommand{\ip}{\ensuremath{\mathsf{IP}}}
\newcommand{\qip}{\ensuremath{\mathsf{QIP}}}
\newcommand{\qiptwo}{\ensuremath{\mathsf{QIP}(2)}}
\newcommand{\qma}{\ensuremath{\mathsf{QMA}}}
\newcommand{\pspace}{\ensuremath{\mathsf{PSPACE}}}
\newcommand{\bqp}{\ensuremath{\mathsf{BQP}}}
\newcommand{\bqsm}{BQSM}

\newcommand{\dfss}{\textsc{dfss-bc}}
\newcommand{\dfrssot}{\textsc{dfrss-ot}}
\newcommand{\ourbc}{\textsc{abo-bc}}
\newcommand{\pfs}{\ensuremath{\mathsf{RR}}}

\newcommand{\hyb}{\ensuremath{\mathsf{Hyb}}}
\newcommand{\nizkot}{\ensuremath{{\sf NIP}}}

\newcommand{\committer}{\ensuremath{\mathsf{C}}}

\title{ The Round Complexity of Proofs in the Bounded Quantum Storage Model}
\author{Alex B. Grilo\inst{1} \and Philippe Lamontagne\inst{2,3}}
\institute{Sorbonne Universit\'e, CNRS, LIP6, France \and National Research Council Canada \and Université de Montréal}
\date{\today}

\begin{document}

\maketitle

\begin{abstract}
The round complexity of interactive proof systems is a key question of practical
  and theoretical relevance in complexity theory and cryptography. Moreover,
  results such as QIP = QIP(3) (STOC'00) show that quantum resources 
  significantly help in such a task.

  In this work, we initiate the study of round compression of
  protocols in the {\em bounded quantum storage model} (BQSM). In
this model, the malicious parties have a bounded quantum memory
and they cannot store the all the qubits that are transmitted in the
protocol.

  Our main results in this setting are the following:

  \begin{enumerate}
  \item There is a non-interactive (statistical) witness indistinguishable
    proof for any language in NP (and even QMA) in BQSM in the {\em plain model}. We notice that in this
    protocol, only the memory of the verifier is bounded.
  \item Any classical proof system can be compressed in a two-message
    quantum proof system in BQSM. Moreover, if the original proof system is
    zero-knowledge, the quantum protocol is zero-knowledge too. In this result,
    we assume that the prover has bounded memory.
\end{enumerate}
Finally, we give evidence towards the ``tightness'' of our results. First, we
show that NIZK in the plain model against BQS adversaries is unlikely with
standard techniques. Second, we prove that without the BQS model there is no
2--message zero-knowledge quantum interactive proof, even under computational
assumptions.

\end{abstract}

\section{Introduction}
\label{sec:intro}

The round complexity of interactive proof systems\footnote{In an interactive proof system,  an all-powerful prover wants
convince a computationally bounded verifier that $x\in L$ for some language $L$
by exchanging polynomially many messages.  We want such interactive protocols
such that the prover can convince the verifier if $x \in L$, whereas if $x\notin
L$ the prover cannot convince the verifier except with negligible probability}
 is a central question in
complexity theory and cryptography. For example, while it is expected that not
all interactive proof systems can be compressed to a constant number of
rounds, showing such a result would have major implications in complexity theory
such as P $\ne$ PSPACE. In cryptographic settings, the round complexity is very
relevant to the practical applications of protocols, specially in the
setting of zero-knowledge (ZK) proof systems~\footnote{In a zero-knowledge interactive
proof system, the verifier ``learns nothing'' from the interaction with the
prover. This is formally defined as requiring the existence of a simulator which
can produce the same output as the verifier, but without the help of the prover.
Zero-knowledge proofs are extremely useful in building other cryptographic
primitives, such as a maliciously secure multiparty computation~\cite{4568209},
IND-CCA encryption~\cite{blumNoninteractiveZeroknowledgeIts1988}, identification
and digital signatures schemes~\cite{fiat_how_1987}.}. While there exist
$4$-messages ZK protocols for \npol{}~\cite{feigeZeroKnowledgeProofs1990}, it is known that $2$-messages ZK protocols for
\npol{} are impossible~\cite{goldreich_definitions_1994}; and in specific
settings such as black-box zero-knowledge, even $3$-message private-coin protocols and constant-round
public-coin protocols are known to be impossible~\cite{goldreichCompositionZeroKnowledgeProof1996}.
These negative results on the round complexity can often be circumvented through
additional resources.
For example, in the
random oracle model, any public-coin zero-knowledge proof can be made
non-interactive through the use of the Fiat-Shamir
heuristic~\cite{fiat_how_1987}. While it has been shown that such a heuristic
cannot be implemented in a black-box
way~\cite{GoldwasserK03,bitansky2013fiat}, it is possible to instantiate it
in specific settings and achieve, for example, non-interactive ZK for \npol{} in
the common reference string (CRS)
model from standard
cryptographic assumptions~\cite{blumNoninteractiveZeroknowledgeIts1988,PeikertS19}. 

With the development of quantum computing, the notion of interactive protocols
has been also extended to the quantum setting. Here, the prover and
verifier are now allowed to exchange quantum messages back-and-forth to prove
that $x \in L$. One of the first results in this direction already indicated
that quantum resources are useful in reducing the rounds of protocols: it was
shown that any quantum interactive protocols can be compressed to a $3$-message
protocol~\cite{KitaevWatrous}. The natural question
raised by such a result is the power of two-messages quantum interactive proof
systems. More concretely, can we compress any quantum (or less ambitiously
classical) protocol into a one-round protocol with quantum communication? This
is tightly connected with the question of instantiation of Fiat-Shamir with
quantum resources which was recently shown black-box impossible in~\cite{dupuis_fiat-shamir_2022}. 

In this work, we make progress in this direction by studying non-interactive and
compression of protocols in the {\em bounded quantum storage model} (BQSM). In
this model, we assume that the malicious parties have a bounded quantum memory
and that they cannot store the all the qubits that are transmitted in the
protocol. We notice that in our protocols, the honest parties do not need
quantum memory at all: they measure the qubits as soon as they are received.
This model has been shown very powerful, allowing the
implementation of several important cryptographic primitives with
information-theoretic
security~\cite{doi:10.1137/060651343,damgard_tight_2007,damga_ard_secure_2007,barhoush_powerful_2023}.
In this work, we show that such a powerful
tool is also relevant for round-efficient interactive protocols.
More concretely, we show the following:
\begin{enumerate}
  \item There is a non-interactive (statistical) witness indistinguishable
    proof for any language in \npol{} (and even \qma{}) in the {\em plain model} against
    BQS adversaries. We notice that in this
    protocol, only the memory of the verifier is bounded.
  \item Any classical proof system can be compressed in a two-message
    quantum proof system in BQSM. Moreover, if the original proof system is
    zero-knowledge, the quantum protocol is zero-knowledge too. In this result,
    we assume that the prover has bounded memory.
\end{enumerate}

We present now our results in more detail and give a brief overview on the
techniques to prove them.

\subsection{Our Results}
\label{sec:results}

As previously mentioned, in this work, we investigate the round complexity of proof systems in the bounded quantum
storage (BQS) model. It is based on the physical assumption that the adversary
has a bounded-size quantum memory of $q(\lambda)$ qubits where $\lambda$ is the
security parameter. Our main results are two compilers for reducing the round complexity of proofs
in the BQS model. Each one operates differently and has its own
applications. The bounded party differs in each of our main results; either the
verifier or the prover has bounded quantum memory, but never both. The memory
bound $q$ on the malicious party is independent of the underlying proof system
and a larger bound can be tolerated by increasing the size of the quantum
messages.

\paragraph{Non-interactive proof for \npol.} 
In our first result, we provide a compiler $\nizkot$ that takes a
$3$--message public-coin interactive proof system with 1-bit challenges and
compresses it to one message.
The main idea of the compiler is to use non-interactive oblivious transfer (OT) in non-interactive
proofs, an idea which  was introduced by~\cite{kilian_minimum_1990} and first appeared in writing in~\cite{bellare_non-interactive_1990}.

More concretely, the starting point of our protocol is the non-interactive quantum oblivious transfer protocol
of~\cite{damgard_tight_2007} which is secure against BQS receivers. We can
construct a non-interactive proof by having the prover send its first
message\footnote{The first message may be classical in a $\Sigma$--protocol or
  quantum, in which case we call it a $\Xi$--protocol~\cite{BG22}.} $a$ in the
clear and input the responses $r_0,r_1$ to both possible challenges $c\in\bool$
as its inputs to OT. Our compiler preserves the soundness of the underlying
interactive proof, and it can be amplified through parallel repetition.
Intuitively, the security of BQS-OT implies that a quantum memory bounded
verifier will only receive one of the two transcripts, which reveals no
information since accepting transcripts can be simulated if the underlying
$\Sigma$-protocol is honest-verifier zero-knowledge.

While we manage to prove that the protocol satisfies the witness
indistinguishable property, proving zero-knowledge is challenging since
it is hard for the simulator to ``decode'' the measurements of a potentially
malicious verifier. 
In particular, we prove in Section~\ref{sec:hvzk-wh-wi} that a ``natural'' simulation technique
cannot work.

\begin{result}
  Let $\Pi$ be a $\Sigma$--protocol. Then $\nizkot[\Pi]$ preserves
  soundness and preserves witness indistinguishability against BQS verifiers.
\end{result}

Our compiler can be extended in a trivial way to $\Sigma$--protocols with
logarithmic challenge length (by using a $1$-out-of-$2^{p}$ OT with $p\in
O(\lg(\lambda))$). Furthermore, the first message of the prover may be quantum,
so our compiler can be applied to $\Xi$--protocols as long as the
verifier is receive-and-measure. Our result thus implies a NIWI for \qma{} based
on the $\Xi$--protocol from~\cite{BG22} which has short challenges and is
receive-and-measure for the verifier.

This compiler allows us to achieve a non-interactive (statistically) witness
indistinguishable proof for all languages in $\npol$ in the BQS model without any prior setup.

\begin{result}\label{re:niwi}
  For any $L\in \npol$, there is a quantum non-interactive proof system for $L$
  with unconditional soundness and witness indistinguishability against BQS
  verifiers.
\end{result}

To obtain Result~\ref{re:niwi}, we apply our compiler to the typical proof
system for the $\npol$--complete language of graph Hamiltonicity. This would
normally introduce a computational assumption on either the prover or the
verifier since the proof uses a commitment scheme, however we can instead use a
quantum bit commitment, which only needs to satisfy a very weak notion of
binding.

A stronger notion than witness indistinguishability (yet still weaker than
zero-knowledge) is witness hiding. We show that witness hiding can be preserved
by our compiler in a regime where the soundness error is inverse polynomial. See
Appendix~\ref{sec:witn-hiding} for details.
 
\paragraph{A Round Collapse Theorem in the \bqsm.}
We show that under the BQS assumption, the round complexity of proof systems
essentially collapses to two messages (one round). We present a \emph{round
  reduction} compiler $\pfs$ that takes as input a $\poly[\lambda]$ rounds
 interactive proof $\Pi$ and produces a single round ($2$ messages)
proof for the same language with the following properties.
\begin{result}
  Let $\Pi$ be a $\poly[\lambda]$--round interactive proof system,
  then there is a $1$--round quantum interactive proof $\pfs[\Pi]$ such that
  \begin{enumerate}
  \item soundness is preserved against BQS provers;
  \item zero-knowledge is preserved.
  \end{enumerate}
\end{result}
Our compiler $\pfs$ is conceptually very simple. It relies on a distinctive
property of the original bit commitment in the \bqsm{}, in that the committer commits to a
bit $b$ by measuring a state it gets from the receiver. This allows us to remove
one round of interaction by having the verifier send a state $\ket\psi$ for the
commitment at the same time as its next challenge $c_i$. The prover commits to
its message $a_i$ by measuring $\ket \psi$, then receives $c_i$, and can respond
with its next message $a_{i+1}$. Since the prover has bounded quantum memory,
it will have to perform a (partial) measurement on $\ket \psi$ before
receiving the verifier's challenge $c_i$. By the binding property of the
commitment against BQS provers, this implies that $a_i$ is independent of $c_i$,
and thus any attack against this protocol is also an attack against $\Pi$. By
repeating this technique for every round in protocol $\Pi$, we end up with a
protocol with one round that has the same soundness error, plus a negligible
term from the BQS-BC binding theorem.

By using the correspondence $\ip=\pspace$~\cite{Shamir}, we obtain the following.

\begin{result}
  $\pspace = \qiptwo^{\mathsf{BSQM}}$, i.e., there exists a $2$--message quantum
  protocol for every problem in \pspace{} if the computationally unbounded
  prover has a bounded quantum memory.
\end{result}

Furthermore applying our compiler to the doubly efficient protocols for
delegation of classical computation~\cite{GKR15,RRR21}, we achieve the following
application.

\begin{result}
  In the BQSM, there is a quantum interactive protocol for any language in
  $\mathsf{P}$ such that the honest prover runs in polynomial time, the verifier
  runs in linear time and logarithmic space, and there is a single round of
  communication.
\end{result}

By applying our compiler to a concrete scheme, we get the first $1$--round
interactive proof for $\npol$ that is both statistically sound (against BQS
provers) and statistically ZK against arbitrary verifiers.

\paragraph{Other Contributions.}

We give evidence towards the ``tightness'' of our results. We show that NIZK in
the plain model against BQS adversaries is unlikely with standard techniques. We
also show that an assumption such as the BQSM is necessary for our round
compression result by proving that there is no 2--message zero-knowledge quantum
interactive proof system when the prover is not memory-bounded. This result is
an extension of the impossibility of~\textcite{goldreich_definitions_1994} to
the quantum case and is presented in Appendix~\ref{sec:impo-aux}.

Our round reduction transform uses a string commitment built by parallel
composition of the original BQS-BC scheme. To commit to $n(\lambda)\in
O(\lambda)$ bit strings requires sending $\lambda\cdot n(\lambda)\in
O(\lambda^2)$ qubits against a $O(\lambda)$--bounded adversary. Thus, the memory
bound is sublinear in the number of transmitted qubits. In
Appendix~\ref{sec:new-bqsm-bc}, we propose a new string commitment where the
length of committed strings, the number of transmitted qubits and the memory
bound are all linear in $\lambda$. While we were unable to prove that this new
commitment meets the definition of binding required by our $\pfs$ transform, we
can show that it is sum-binding, so it might be useful in improving the
efficiency of other BQSM schemes.

\subsection{Related Work}
\label{sec:related}

Classical non-interactive witness indistinguishable proof systems can be built
from strong computational assumptions such as a derandomization circuit
complexity
assumption~\cite{barakDerandomizationCryptography2003,bitanskyZAPsNonInteractiveWitness2015}
and the decision linear assumption on bilinear
groups~\cite{grothNoninteractiveZapsNew2006}.

Quantum NIZK for \qma{} can be achieved in the following models: in the secret
parameter model~\cite{broadbent_qma-hardness_2020}, in the QROM with quantum
preprocessing~\cite{morimaeClassicallyVerifiableNIZK2022}, in the designated
verifier model~\cite{shmueliMultitheoremDesignatedVerifierNIZK2021}, using
pre-shared EPR pairs and subexponential
assumptions~\cite{bartusekSecureComputationShared2023a}, and with a CRS with an
instance-dependent quantum message from the verifier to the
prover~\cite{coladangelo_non-interactive_2020}. While we call the bounded
quantum storage assumption a ``model'', our results do not rely on any prior
setup.

The bounded quantum storage model was introduced in~\cite{doi:10.1137/060651343}
as a physical assumption upon which information theoretically secure two-party
cryptographic primitives such as oblivious transfer (OT) and bit-commitment (BC)
could be built. The BQSM has found further application to quantum key
distribution~\cite{damgard_tight_2007,damga_ard_secure_2007} and to secure
identification~\cite{damga_ard_secure_2007}. The noisy quantum storage
model~\cite{wehner_cryptography_2008,schaffner_robust_2009,konig_unconditional_2012}
(NQSM) is a generalization of the BQSM, where the adversary's quantum memory is
subject to noise, that enables protocols for OT and BC. There are OT protocols
in the BQSM and NQSM where the tolerated bound or noise level is an arbitrary
large fraction of the number of exchanged
qubits~\cite{dupuis_entanglement_2015}. The model was recently exploited to
achieve strong primitives such as one-time
programs~\cite{barhoush_powerful_2023}. Composability frameworks have been
proposed for the
BQSM~\cite{unruh_concurrent_2011,wehner_composable_2008,fehr_composing_2009}.
These results and that of~\cite{barhoush_powerful_2023} require extracting the
malicious party's input, which in general is inefficient. This is not a problem
when the class of adversary is quantum memory bounded and computationally
unbounded, but it doesn't work when simulation needs to be efficient, as in ZK
proofs. Finally, post-quantum zero-knowledge  against BQS adversaries was
recently studied~\cite{ananthPostQuantumZeroKnowledgeSpaceBounded2022} in the
context where all information exchanged by the parties
is classical, but the adversaries may be quantum.

\section{Preliminaries}
\label{sec:prelims}

For a set $\mathcal{S}$ we write $2^\mathcal{S}$ to denote the powerset, or set
of subsets, of $\mathcal{S}$. We let $\Delta:\bool^n\times\bool^n\rightarrow
[0,1]$ denote the relative Hamming distance between two $n$--bit strings. It is a
well-known fact that for any $x\in\bool^n$, $|\{x':\Delta(x,x')<\delta\}|\leq
2^{H(\delta)n}$ where $H$ is the binary Shannon entropy.

We use ``+'' and ``$\times$'' to refer to the computational and Hadamard bases.
We often specify the base using a bit and write $\ket x_\theta:=H^{\theta}\ket
x$ for $x,\theta\in\bool$ where $H$ is the Hadamard transform.

Throughout this paper, $\|\cdot\|$ denotes the trace norm
$\|A\|=\trace{\sqrt{A^*A}}$ when its argument is an operator and the Euclidean
norm $\|\ket \psi\| = \sqrt{\braket{\psi}{\psi}}$ when its argument is a vector.

If $U$ is an isometry and $\rho$ a mixed state, we write $U(\rho):= U\rho U^*$.
If $\ket\psi$ is pure, we sometimes write $\psi$ as shorthand for the mixed
state $\proj \psi$.

For a string $a\in\bool^n$, we let $a_i^j$ for $1\leq i<j\leq n$ denote the
substring of $a$ composed of the bits $a_i,\dots,a_j$. When $a_1,\dots,a_k$ are
Boolean strings, we let $(a_i)_j$ denote the $j$th bit of the $i$th string.

The
\emph{min-entropy} of a classical random variable $X$ conditioned on an event
$\Psi$ is $H_\infty(X|\Psi)=-\log \max_x \Pr[X=x|\Psi]$. The \emph{max-entropy}
of a quantum or classical register $A$ in state $\rho$ is $H_0(A)_\rho=
\log\rank{\rho_A}$. A trivial upper-bound on $H_0(A)$ is $\dim A$.
The \emph{min-entropy splitting lemma} will also be useful. For a proof of this
lemma, please refer to the full version of~\cite{damgard_tight_2007}.
\begin{lemma}[Min-entropy splitting]\label{lem:splitting}
  Let $X_0$, $X_1$ and $Z$ be random variables with $H_\infty(X_0X_1|Z)\geq
  \alpha$. Then there exists a random variable $C$ with support over $\bool$
  such that $H_\infty(X_{1-C}|ZC)\geq \alpha/2-1$.
\end{lemma}

\subsection{The Bounded Quantum Storage Model}
\label{sec:bqsm}

In the BQSM, the adversary can act arbitrarily on the whole state before a point
at which the memory bound is
applied~\cite{doi:10.1137/060651343,damga_ard_secure_2007,damgard_tight_2007}.
In particular, it is allowed an arbitrary CPTP map that transforms its $\lambda$
qubits into $q(\lambda)<\lambda$ qubits plus unlimited classical
side-information. This memory bound is applied at one or many points in the
protocol. Between memory bounds, the adversary again becomes unbounded and may
use arbitrarily-many auxiliary qubits to aid in its computation. In this paper,
the memory bounds apply only to a malicious party, as the protocols are
prepare-and-measure, which requires no quantum memory.

We review two important protocols in the \bqsm{} and state their security
properties below. 

\subsubsection{Bit Commitment in the \bqsm{}.}
\label{sec:bc-bqsm}

We begin by discussing the bit commitment scheme
of~\cite{doi:10.1137/060651343}. One of the unique features of this protocol is
that the \emph{commiter}  commits to a bit through the measurement
of a received quantum state, and does not need to send any message back to the receiver of the
commitment. The opening phase consists of the transmission of the
committed bit and along with measurement outcomey, which will enable the
consistency verification. This commitment scheme is perfectly hiding since no information is
sent to the receiver prior to the reveal phase. The
original~\cite{doi:10.1137/060651343} bit commitment protocol in the \bqsm{}
proceeds is described below.

\begin{quote}
  \rule{\linewidth}{1pt}
  \begin{center}
    {\bf Protocol} $\dfss{}$
  \end{center}

  {\bf Input:} a bit $b\in\bool$ for the committer.

  {\bf Commit phase:}
  \begin{enumerate}
  \item \verifier{} sends $\ket{x}_\theta$ for $x\in\bool^n$ and
    $\theta\in\{+,\times\}^n$ to the committer.
  \item \committer{} commits to bit $b$ by measuring all qubits in basis $+$ if
    $b=0$ and in basis $\times$ if $b=1$, obtaining a measurement outcome $x'$. 
  \end{enumerate}
  
  \begin{center}
    \emph{--- Memory bound applies ---}
  \end{center}

  {\bf Reveal phase:}
  \begin{enumerate}
  \item To open the commitment, \committer{} sends $b$ and $x'$ to \verifier{} who
    checks that $x'_i=x_i$ whenever $\theta_i=b$.
  \end{enumerate}

  \rule{\linewidth}{1pt}
\end{quote}

\paragraph{Binding Property of BC in the \bqsm{}.}
\label{sec:binding}

Since unconditionally secure bit commitment is impossible in the quantum
setting~\cite{loWhyQuantumBit1998,mayersUnconditionallySecureQuantum1997},
binding relies on the quantum storage bound of the malicious committer. A
malicious committer $\tilde\committer$ is bound to a single value by the fact
that it is forced to perform a partial measurement on the register it receives.
This notion is formalized by the following definition.
Let
\begin{equation}
  \label{eq:binding-state}
  \rho_{EWV} = \sum_{w,v} P_{WV}(w,v)\cdot \rho_{E}^{w,v}\otimes \proj{w} \otimes \proj{v}
\end{equation}
be the joint state at the end of the commit phase after the memory bound is
applied, where $W$ is $\tilde\committer$'s classical register, $E$ is
$\tilde\committer$'s $q$--qubit quantum register and $V$ is the receiver's
state.

\begin{definition}\label{def:binding-DFRSS} 
  A commitment scheme in the bounded-quantum-storage model is called {\em
    $\epsilon$-binding}, if for every (dishonest) committer
  $\tilde{\committer}$, inducing a joint state $\rho_{EWV}$ after the commit
  phase, there exists a classical random variable $B'$ with support in
  $\bool^n$, given by its conditional distribution $P_{B'|WV}$, such that for
  any $b'\in\bool^n$, the state
  \begin{equation}
    \label{eq:rho-b} \rho_{EWV}^{b'} = \sum_{w,v} P_{WV|B'}(w,v|b) \cdot \rho_{E}^{w,v}\otimes\proj{w} \otimes \proj{v} 
  \end{equation} satisfies the following condition. When executing the opening
  phase on the state $\rho_{EWV}^{b'}$, for any strategy of $\tilde\committer$, the
  honest verifier accepts an opening to $b\neq b'$ with probability at most $\epsilon$.
\end{definition}

It was shown in~\cite{damgard_tight_2007} that \dfss{} satisfies the above
definition (for $b\in\bool$). This implies a string commitment protocol where
$\committer$ commits bit-wise to $b_i$ using protocol $\dfss$.

\begin{theorem}[Security of DFSS-BC]\label{thm:binding-dfss} The quantum bit commitment scheme \dfss{} is binding according to
  Definition~\ref{def:binding-DFRSS} against $q$--bounded committers if
  $n/4-q\in\Omega(n)$.
\end{theorem}

\subsubsection{Oblivious Transfer in the BQSM.}
\label{sec:ot-bqsm}

The original OT protocol in the \bqsm{} was a Rabin OT (where the sender has one
input and the receiver gets to see it with probability $\frac 12$). We use the
$\binom 21$--OT from~\cite{damgard_tight_2007} which is presented below. It is a
\emph{non-interactive} protocol which consists of a single message with quantum
and classical parts from the sender to the receiver. The memory bound is applied
after the transmission of the quantum state.

\begin{quote}
  \rule{\linewidth}{1pt}
  \begin{center}
    {\bf Protocol} \dfrssot{}
  \end{center}

  {\bf Input:} two bits $s_0,s_1\in\bool$ for the sender. A bit $c\in\bool$ for
  the receiver.

  {\bf Sender:}
  \begin{itemize}
  \item Prepare $\ket x_\theta$ and send it to the receiver.
  \end{itemize}
  \begin{center}
    \emph{--- Memory bound applies ---}
  \end{center}
  \begin{itemize}
  \item Pick two universal hash functions $h_0,h_1\in\mathcal{H}$ and
    set $m_0=s_0\oplus h_0(x_0)$ and $m_1=s_1\oplus
    h_1(x_1)$ where $x_0$ (resp. $x_1$) is the substring of $x$ for which
    $\theta_i=+$ (resp. $\times$).
  \item Send $(\theta,h_0,h_1,m_0,m_1)$ to the receiver.
  \end{itemize}    

  {\bf Receiver:}
  \begin{itemize}
  \item Measure each qubit of the quantum state in basis $[+,\times]_c$ to get a
    result $x'$.
  \item Compute $x'_c$ using $\theta$ and output $m_c\oplus h_c(x'_c)$.
  \end{itemize}
  \rule{\linewidth}{1pt}
\end{quote}

Correctness of the protocol follows from the fact that $x'_c=x_c$ if both
parties follow the protocol. The security is established by the following
result. 

\begin{theorem}[Security of DFRSS-OT~\cite{damgard_tight_2007}]\label{thm:bqs-ot}Let $R$ be a malicious $q$-bounded receiver against $\ell$--bit \dfrssot{} and let
  $\rho_{M_0 M_1 H_0 H_1 E}$ be the state of $R$ right after the classical
  message from the sender (where $\dim E\leq 2^q$). Then there exists a random
  variable $C$ such that
  \begin{equation}
    \label{eq:security-ot}
    \left\|
      \rho_{M_{1-C} M_C C H_0 H_1 E} -
      \frac{\id_{M_{1-C}}}{2^{\ell}} \otimes \rho_{M_C C H_0 H_1 E}
    \right\|
    \leq 2^{-\frac n4 +\ell + q}
  \end{equation}
\end{theorem}

\paragraph{Parallel repetition of \dfrssot.}
While protocol \dfrssot{} does not generally compose in
parallel\footnote{See~\cite{wehner_composable_2008} for a counter-example. The
  issue occurs when the same party acts as the receiver of an OT instance while
  simultaneously acting as the sender in another. }, it does compose in the case
where the same party is the sender in every instance. By parallel repetition, we mean
the protocol where the quantum part of every instance is sent before the memory
bound, followed by the classical part of every instance.   
\begin{corollary}[Parallel repetition of DFRSS-OT]\label{thm:bqs-ot}Let $R$ be a malicious $q$-bounded receiver against $k$ parallel repetitions
  of $\ell$--bit \dfrssot{}. Let $\rho_{\vec M_0 \vec M_1 \vec H_0 \vec H_1 E}$
  be the state of $R$ right after the classical message from the sender (where
  $\dim E\leq 2^q$). Then there exist random variables $\vec C=C^1,\dots,C^k$ such that
  \begin{equation}\label{eq:security-parallel-ot}
    \left\|
      \rho_{\vec M_{\neg \vec C}\vec M_{\vec C} \vec C\vec H_0\vec H_1E} -
      \frac{\id_{\vec M_{\neg\vec C}}}{2^{k\cdot\ell}} \otimes
      \trace[\vec M_{\neg \vec C}]{\rho_{\vec M_{\vec C} \vec C \vec  H_0 \vec H_1 E}} 
    \right\|
    \leq k\cdot 2^{-\frac n4 +\ell + q}
  \end{equation}
  where $\vec M_{\neg \vec C}$ denotes registers $M^i_{1-C^i}$ for $i\in [k]$
  and $M_{\vec C}$ denotes registers $M^i_{C^i}$.
\end{corollary}

\begin{proof}
  Consider the purified variant of the scheme, where the sender sends halves of
  EPR pairs in the first step and measures its halves in basis $\theta$ after
  the memory bound. Consider $k$ parallel executions of this purified scheme.
  Let $X^i$ and $\Theta^i$ be the measurement result and basis for the $i$th
  repetition. The distribution $(X^i,\Theta^i)$ is independent from that of
  $(X^j,\Theta^j)$ for $j\neq i$. Since $H_\infty(X^i| \Theta^i)\geq (\frac
  12-\epsilon)n$, by the min-entropy splitting lemma there exists $C^i$ such
  that $M_{1-C^i}$ is indistinguishable from uniform. Note that the min-entropy
  bound holds even if we condition on the random variables from other executions
  and on the receiver's registers:
  \begin{equation*}
    H_\infty(X^i| (\Theta^j)_j (X^{j})_{j\neq i}ZE)\geq   H_\infty(X^i| \Theta^i)-H_0(E)\geq (\frac 12-\epsilon)n -q
  \end{equation*}
  Also note that the random variable $C^i$ depends only on the conditional
  distribution $P_{X^i|\Theta^i}$, so the $C^i$s are simultaneously well-defined
  for each $i\in[k]$. We have that for each $i$,
  \begin{equation}
    \left\|
      \rho_{M^i_{1- C^i}M^i_{C^i} C^iH^i_0H^i_1E} -
      \frac{\id_{M^i_{1- C^i}}}{2^{k\cdot\ell}} \otimes \trace[M^i_{1-C^i}]{\rho_{M^i_{C^i} C^i H^i_0 H^i_1 E}}
    \right\|
    \leq  2^{-\frac n4 +\ell + q}
  \end{equation}
  and, by starting with $\rho_{\vec M_{\neg \vec C}\vec M_{\vec C} \vec C\vec
    H_0\vec H_1E}$ and invoking the triangle inequality $k$ times (where each
  time we replace $M^i_{1-C^i}$ with the completely mixed state), we get the
  corollary's statement. \qed
\end{proof}

\subsection{Quantum Interactive Proofs and Quantum Zero-Knowledge}
\label{sec:qzk}

An interactive proof system is a protocol between two participants, a prover
\prover{} and a verifier \verifier{}. We consider proofs of language membership
where each participant receives a common input $x$, and the prover may receive
an additional input $w$, such as a witness that $x$ is a member of a $\npol$
language. A proof system is \emph{classical} if the message exchanged are
classical, but \prover{} and \verifier{} are allowed to be quantum. We say that
a classical or quantum proof system is \emph{public coin} if the verifier's
messages are uniformly and independently distributed.

We denote by $\prover(x)\leftrightharpoons \verifier(x)$ the output of the
verifier after the interactive proof. An interactive proof system for a language
$L$ is $\delta$--correct if for all $x\in L$,
\begin{equation}
  \label{eq:correctness}
  \Pr[\prover(x)\leftrightharpoons
  \verifier(x)=1]\geq \delta\enspace .
\end{equation}
 It is (computationally) $\epsilon$--sound if for all
(\qpt) malicious prover $\tilde \prover$, for all $x\notin L$,
\begin{equation}
  \label{eq:soundness}
  \Pr[\tilde\prover(x)\leftrightharpoons \verifier(x)=1]\leq \epsilon \enspace .
\end{equation}

We now define quantum zero-knowledge~\cite{watrous_zero-knowledge_2009}.

\begin{definition}[Indistinguishability of Quantum States]
  Let $L$ be an infinite set of strings and let $\psi=\{\psi_x\}_{x\in L}$ and
  $\phi=\{\phi_x\}_{x\in L}$ be two families of quantum states. We say that
  $\psi$ and $\phi$ are \emph{computationally indistinguishable} if for all
  $x\in L$ for every $\poly[|x|]$--size quantum circuit $\dist$ and for all
  state $\sigma$ over $\hilbert^{\otimes \poly[|x|]}$,
  \begin{equation*}
    \| \dist(\psi_x\otimes \sigma) - \dist(\phi_x\otimes \sigma)\| \leq \negl[|x|]\enspace .
  \end{equation*}
  $\psi$ and $\phi$ are \emph{statistically indistinguishable} the above holds
  with respect to all $\dist$ and all states $\sigma$.
\end{definition}

\begin{definition}[Indistinguishability of Quantum Channels]
  \label{def:CPTP-indist}
  Let $L$ be an infinite set of strings and let $\Psi=\{\Psi_x\}_{x\in L}$ and
  $\Phi=\{\Phi_x\}_{x\in L}$ be two families of superoperators agreeing on their
  input and output spaces: $\Psi_x,\Phi_x:\hilbert^{\otimes
    n(|x|)}\rightarrow\hilbert^{\otimes m(|x|)}$. We say that $\Psi$ and $\Phi$
  are \emph{computationally indistinguishable} if for all $x\in L$, for
  every $\poly[|x|]$--size measurement circuit $\dist:\hilbert^{\otimes
    m(|x|)+k(|x|)} \rightarrow \hilbert$ and for every
  $\sigma\in\hilbert^{\otimes m(|x|)+k(|x|)}$,
  \begin{equation}
    \label{eq:CPTP-indist}
    \|\dist (\Psi_x\otimes \id^{\otimes k(|x|)}(\sigma)) -
    \dist (\Phi_x\otimes \id^{\otimes k(|x|)}(\sigma)) \|
    \leq \negl[|x|]
  \end{equation}
  where $m(|x|)$, $n(|x|)$ and $k(|x|)$ are $\poly[|x|]$. $\Psi$
  and $\Phi$ are \emph{statistically indistinguishable} if the above holds with
  respect to all CPTP map $\dist$ and all states $\sigma$ (for unbounded
  $k(|x|)$).
\end{definition}

\begin{definition}[Quantum Zero-Knowledge]\label{def:t-zk}
An interactive proof system $\Pi=\langle \prover, \verifier\rangle$ for a
  language $L$ is \emph{computationally quantum zero-knowledge} (qZK) if for
  every $\poly[|x|]$--time verifier $\verifier^*$ receiving the common input
  $x\in L$ and a $\poly[|x|]$--size quantum register $E$, there exists a
  $\poly[|x|]$--time simulator $\simulator_{\verifier^*}$ that receives the same
  inputs and such that the quantum channel families $\{\prover
  \leftrightharpoons \verifier(x,\cdot))\}_{x\in L}$ and
  $\{\simulator_{\verifier^*}(x,\cdot)\}_{x\in L}$ are computationally
  indistinguishable. We say that $\Pi$ is \emph{statistically quantum
    zero-knowledge} $\verifier^*$ is unbounded and if the two channel families
  are statistically indistinguishable. We say it is (computationally or
  statistically) \emph{quantum honest verifier zero-knowledge} (qHVZK) if
  indistinguishability holds with respect to the honest verifier
  $\verifier^*=\verifier$.
\end{definition}

\begin{definition}[$\Xi$--protocols]\label{def:xi-protocol}
  A $\Xi$--protocol for a language $L$ is an interactive proof system
  $\Pi=(\prover_1,\prover_2,\verifier)$ with the following structure.
  \begin{enumerate}
  \item The prover receives as input $x$ and a witness $\ket{w}$, computes
    $\ket \phi_{AB}\leftarrow \prover_1(x)$ and sends $\phi_B$ to the verifier.
  \item The verifier chooses a uniformly random challenge $c\in\bool^\ell$ and
    sends $c$ to the prover.
  \item The prover computes $r\leftarrow\prover_2(x,\phi_B,c)$ and sends $r$ to
    the verifier.
  \item The verifier accepts if $\verifier(x,\phi_A,c,r)=1$ and rejects
    otherwise.
  \end{enumerate}
  A $\Sigma$--protocol is a $\Xi$--protocol where $\ket w$ and $\ket\phi_{AB}$
  are classical. A $\Xi$--protocol is \emph{prepare-and-measure} for the
  verifier if the verifier measures $\phi_A$ upon reception in a basis chosen by
  $c$ and the predicate $\verifier$ is applied on the measurement outcome.
\end{definition}

\section{Non-Interactive Proofs in the BQSM}
\label{sec:nizk-ot}

We present a generic transform to turn arbitrary $\Sigma$--protocols with small
challenge space to non-interactive proofs. We actually consider a slight
generalization of $\Sigma$--protocols where the first message send by the prover
can be a quantum state, while the challenge by the verifier should be uniformly
random bits and the third message by the prover is classical. \textcite{BG22} called
this type of protocols as $\Xi$--protocols (Definition~\ref{def:xi-protocol}), and we will use their notation to
stress that the first message can be quantum.

The soundness of our transform does not rely on any setup assumption.
We will show later that while we cannot show zero-knowledge for such a transform,
we can prove some weaker notions. We notice that  since we are working in the
bounded storage model, we consider $\Xi$ protocols where an honest verifier
measures the qubits of the first message as they arrive based on the chosen
challenge.

\begin{quote}
  \rule{\linewidth}{1pt}
  \begin{center}
    {\bf Protocol} $\nizkot[\Pi]$ for a $\Xi$--protocol $\Pi$
  \end{center}

  {\bf Prerequisite: } A $3$--message, $1$-bit public coin, interactive proof
  $\Xi=(\prover_1,\prover_2,\verifier)$.

  \vspace{1em}
  {\bf Prover: } 
  \begin{enumerate}
  \item For $i\in[k]$,
    \begin{enumerate}
    \item Prepare the $n$--qubit BB84 state $\ket {x_i}_{\theta_i}$
    \item Compute $\ket{\phi_i}\leftarrow \prover_1$ with fresh randomness each
      time
    \item Compute responses $r_i^c$ using $\prover_2$ for $c\in \{0,1\}$
    \item Sample two universal$_2$ hash functions $h_i^0$ and $h_i^1$
    \item Compute $m_i^0=r_i^0\oplus h_i^0(x^+_i)$ and $m_i^1=r_i^1\oplus
      h_i^1(x^\times_i)$
    \end{enumerate}
  \item Send $\bigotimes_i \ket{\phi_i}\ket{x_i}_{\theta_i}$ to the verifier
  \end{enumerate}    
  \begin{center}
      \emph{--- Memory bound applies ---}
    \end{center}
    \begin{enumerate}[resume]
  \item Send
    $\bigotimes_i\ket{\theta_i,h_i^0,h_i^1,m_i^{0},m_i^{1}}$
    to the verifier
  \end{enumerate}

  \vspace{1em}
  {\bf Verifier: } 
  \begin{enumerate}[resume]
  \item Pick a $k$ random selection bits $c_1,\dots c_k$
  \item For $i\in[k]$,
    \begin{enumerate}
    \item Measure $\ket{x_i}_{\theta_i}$ on basis $c_i$ getting $x'_i$
    \item Measure $\ket{\phi_i}$ according to $\verifier$ to get an outcome
      $a_i$
  \end{enumerate}
  \begin{enumerate}[resume]
  \item Compute $x^{c_i}_i$ from $\theta_i$ and $x'_i$; and compute $r^{c_i}$
    from $x^{c_i}_i$, $h_i^{c_i}$ and $m^{c_i}_i$
  \item Check that for all $i\in [k]$ $\verifier(a_i,c_i,r_i^{c_i})=1$, otherwise abort
  \end{enumerate}

  \end{enumerate}

  \rule{\linewidth}{1pt}
\end{quote}

The soundness of the protocol follows from the fact the prover is oblivious to
which response the verifier has learned. Since BQS-OT is secure against
unbounded senders, the soundness of $\nizkot[\Pi]$ is unconditionally
reducible to the soundness of $\Pi$.

We notice that this technique can be used to compress $\log n$ rounds protocols
with $\log n$ bit challenges by using $\poly$ instances of OT. Let's say for
simplicity that we have a $k$ rounds protocol with $1$ bit challenges with $k\in
O(\log n)$, and have access to a $\binom{2^k}{1}$--OT. Then for each of the
$2^k$ inputs $0\leq j< 2^k$, the prover sends the transcript it would produce if
$j$'s bits were the challenges.  We can extend this to $m=O(\log n)$ bit
challenges by considering $j\in \bool^{2^{m+k}}$.

\begin{theorem}
  Let $\Pi$ be a 1-bit challenge $\Xi$--protocol with soundness $\frac
  12$ against quantum adversaries. Then $\nizkot[\Pi]$ is a sound quantum
  non-interactive proof with soundness error $\frac 1{2^k}$.
\end{theorem}
\begin{proof}
  Let $\adv$ be a malicious prover against $\nizkot[\Pi]$. We construct a
  reduction $\rdv$ against (the $k$-wise parallel repetition of) $\Pi$ that
  has the same success probability as $\adv$. The reduction simulates the OT
  instances while being able to recover both messages sent by $\adv$ by having a
  sufficiently large quantum memory. Thus we reduce to the soundness of $\Pi$
  against quantum adversaries.

  When $\adv$ sends the quantum state $\rho_{A_1X_1...A_kX_k}$ in the first step, where $A_i$ is the
  register that is supposed to have the state  $\ket{\phi_i}$ and $X_i$ is the
  register supposed to have $\ket{x_i}_{\theta_i}$, \rdv\ stores the
  qubits.  When $\adv$ sends the classical message
  $(\theta_i,h_i^0,h_i^1,m_i^{0},m_i^{1},a_i)_{i\in [k]}$, $\rdv$ measures
  the register $X_i$ in basis $\theta_i$ and compute 
the responses to each possible challenge $r^{0}_i$ and $r^{1}_i$.

  Now $\rdv$ acts as the sender in protocol $\Pi^k$. It sends  the state
  $\rho_{A_1...A_k}$ as the first message of the
  $\Pi$ protocol. Upon reception of the
  challenges $c_1,\dots,c_k\in\bool$ from the verifier, $\rdv$ replies with
  $r_1^{c_1},\dots r_k^{c_k}$. 

  It remains to argue that the verifier accepts in
  $\Pi^k$ against $\rdv$ with the same probability that the verifier accepts
  in $\nizkot[\Pi]$ against $\adv$. This follows from the observation that
  $a_i$, $c_i$ and $r^{c_i}_i$ are identically distributed in both cases.
  Therefore, the success probability of $\rdv$ against $\Pi^k$ is exactly
  that of $\adv$ against $\nizkot[\Pi]$. \qed
\end{proof}

\begin{remark}
  \label{remark:oblivious-soundness} 
  We notice that the soundness of $\nizkot[\Pi]$ actually follows
  from a weaker notion of soundness that we call {\em oblivious soundness},
  which intuitively says that the Prover cannot simultaneously answer the
  two challenges. More
  concretely, $\nizkot[\Pi]$ is sound if $\Pi$ has the following property: for
  any no instance $x\notin L$ and first message $\rho$, no prover can create, at the
  same time, a valid answer for $c = 0$ {\em and} a valid answer for $c = 1$.
  More concretely,  for all possible values
  $(r_0,r_1)$
  \begin{equation}
    \label{eq:oblivious-soundness}
    \sup_{M}  \sum_{b\in\bool} \sum_{r_0,r_1}
    \trace{(M_{r_0,r_1}\otimes V(x,b,r_b)) \rho_{AB}}
    \leq 1+ \negl.
  \end{equation}
  where $M_{r_0,r_1}$ consists of a measurement made by the prover to answer
  $r_0$ to the first challenge and $r_1$ to the second challenge.
  While this property is implied by standard soundness, we will
  see a protocol later in this section that only
  satisfy oblivious soundness.
\end{remark}

\subsection{Security Against Malicious Verifier}
\label{sec:hvzk-wh-wi}

We now turn to the security against malicious verifiers of $\nizkot[\Pi]$. A
verifier with an arbitrarily large quantum memory may postpone its measurement
and learn both transcripts of $\Pi$. If for example $\Pi$ is
special-sound, it would allow them to recover an \npol{} witness. Thus, we focus the
security against bounded quantum storage verifiers. The question remains as to exactly
what properties can be proven in this setting. 

In this section, we give evidence that proving zero-knowledge for
$\nizkot[\Pi]$ (or variations of it) might be out of
reach for non-interactive proofs in the \bqsm.  However, we show that this
protocol preserves some properties of the $\Pi$ such as Witness
Indistinguishability and Witness Hiding properties.

\subsubsection{Impossibility of ``Black-Box'' Non-Interactive Zero-Knowledge in
  the BQSM.}
\label{sec:impo-bbzk}

In order to achieve zero-knowledge, one would need to construct a simulator that can
produce an output that is indistinguishable from the output in the real protocol.
For that, the simulator should have minimal access to the verifier's state and
be able to run its program. We show here that only looking at the state of the
verifier after the memory bound is not sufficient to prove zero-knowledge. To
overcome such an impossibility, we would need a ``white-box'' simulator that
take advantage from the {\em code} of the verifier.

We notice that we will show the impossibility result for $\Sigma$ protocols
(i.e. the first message is classical), and that the verifier does not have
access to quantum auxiliary input. These two cases usually makes proving quantum
zero-knowledge much simpler, making our no-go result stronger.

We define an adversarial verifier strategy as a pair of unitaries $V=(V_1,V_2)$
where $V_1$ maps $\ket x_\theta$ and an auxiliary register initialized in state
$\ket 0$ (and potentially an auxiliary quantum input) to registers $E$ and $Z$
where $\dim E\leq 2^q$ and the register $Z$ is measured in the computational
basis when the memory bound applies. The unitary $V_2$ acts on registers $EZ$
and a register $T$ containing the prover's classical transmission, and produces
the verifier's output.

We define a special type of black-box simulator for the $\nizkot$ scheme, which
we call ``BQS-BB'', as a \qpt\ algorithm $\simulator$ that has black-box access
to the unitaries $V_1,V_1^*,V_2$ and $V_2^*$. In particular, the simulator is
allowed to look at the state of the verifier after the memory bound, and can
even purify the verifier's action (i.e.\ without the measurement on $Z$). We
show that this simulation technique cannot be used to prove zero-knowledge.
Intuitively, the reason why simulation is impossible in this setting is that the
simulator cannot (efficiently) retrieve which challenge the verifier could have
information about. This prevents, for instance, 
the simulator to apply the rewinding technique.

This impossibility holds regardless of whether or not the verifier receives an
auxiliary input. 

\begin{lemma}\label{lem:no-zk}
  Let $\Pi$ be an arbitrary 1--bit challenge special-sound $\Sigma$--protocol
  for a hard language $L$ and let $V=(V_1,V_2)$ be an adversarial verifier
  strategy. The non-interactive proof $\nizkot[\Pi]$ is not zero-knowledge
  with BQS-BB simulation.
\end{lemma}
\begin{proof}
  We assume that we are running $\nizkot[\Pi]$ on a single instance of $\Pi$,
  i.e.\ with $k=1$. Let $(\enc,\dec)$ be a symmetric encryption scheme with
  semantic security against quantum
  adversaries~\cite{alagicComputationalSecurityQuantum2016}. We consider a
  family of malicious verifiers that collude with the distinguisher in order to
  thwart any simulation attempt. Let $\{(\dist^k,V^k = (V_1^k, V_2^k))\}_{k\in\bool^\lambda}$
   be
  described as follows.
 
  The isometry $V_1^k$ does the following in a purified manner. 
  \begin{enumerate}
  \item Initialize register $Z=(\Theta,X,P)$ in state $\ket 0_Z$.
  \item Upon reception of an $n$-qubit state $\ket \Psi$, move it to register
    $P$.
  \item Apply $H$ to register $C$ to obtain a uniform
    superposition over $\bool$.
  \item Measure each qubit of register $P$ coherently in basis $C$
    to get $x$, i.e. applied the controlled (by $C$) unitary $\ket
    \psi\mapsto \sum_x\proj x (H^C)^{\otimes n}\ket\psi\ket x$.
  \item Encrypt all registers \emph{in place} using the unitary $\ket
    {m}\mapsto\ket {\enc_k(m)}$ (which is possible since $\enc_k$ is a
    permutation).
\item The state of register $Z$ is now
    \begin{equation}
      \frac 1{\sqrt 2}\sum_{c,x} \bra{x} (H^c)^{\otimes n} \ket \psi \cdot \ket{\enc_k(\theta,x,x)}_Z
    \end{equation}
  \end{enumerate}
  We set the isometry $V_2^k$ as the identity, i.e.\ it just outputs everything
  it receives: the classical memory register $Z$ from $V_1^k$ and the classical
  message $M$ from $\prover$.

  The distinguisher $\dist_k$ receives register $Z$ containing the encryption
  (under key $k$) of the verifier's classical memory and a register $M$
  containing the prover message in protocol $\nizkot[\Pi]$. It decrypts $z$
  to get the verifier's measurement basis $c$ and outcome $x$, uses it to
  recovers one of the two transcripts contained in $M$ and outputs $1$ if the
  transcript obtained is accepting and $0$ otherwise. 
  
  We now argue that it is impossible to simulate such a verifier efficiently.
  First, we notice that in a real execution (where $\verifier$ interacts with $\prover$),
  $\dist_k$ always outputs $1$ assuming perfect correctness of $\Pi$. In a
  simulated execution, we can use the fact that the language is hard (so that
  \simulator\ cannot produce two accepting transcripts) and that the verifier's
  memory is encrypted (so that \simulator\ cannot guess the verifier's
  challenge) to show that the distinguisher outputs $0$ with probability close
  to $\frac 12$.

  By the semantic security of
  $\enc$~\cite{alagicComputationalSecurityQuantum2016}, for any $\qpt$
  $\simulator$ there exists a \qpt\ simulator $\simulator'$ that, whenever
  $\simulator$ calls $V_1^k$ and receives register $Z$
which contains the encryption of the memory of an honest verifier,
   $\simulator'$ ignores register $Z$, but is still able to
  produce an output indistinguishable from $\simulator$.

  \begin{align*}
    &\left| \Pr[\dist_k(\langle \prover(x,w) \rightarrow \verifier(x) \rangle)=1]
      - \Pr[\dist_k(\simulator(x))=1] \right| \\
    &= 1-\Pr[\dist_k(\simulator(x))=1]\\
    &\leq 1- \Pr[\dist_k(\simulator'(x))=1] + \| \simulator'(x)-\simulator(x)\|\\
    &\leq 1- \frac 12 + \Pr[(x,w')\in R_L\mid w'\leftarrow \simulator'(x)] + \negl
  \end{align*}
  In the last inequality, we used the special soundness of $\Pi$ which says
  that producing two accepting transcript for the same commitment $a$ is as hard
  as producing a witness for $x\in L$. By the hardness of $L$, the probability
  of this happening is negligible, and if the output of $\simulator'$ contains
  only one accepting transcripts, then $\dist_k$ outputs $0$ with probability
  $\frac 12$.
  \qed
\end{proof}

Lemma~\ref{lem:no-zk} indicates that techniques restricted to evaluating $V$ and
$V^*$ will not suffice for proving zero-knowledge of $\nizkot[\Pi]$. White-box
techniques exist for ``looking inside'' the verifier to infer the index $\bar c$
of the OT message on which it has uncertainty. See Section~\ref{sec:related} for
examples. These techniques rely on computing the exact probability distributions
induced by the adversary's actions and inferring the random variable $\bar C$
whose existence is established by the min-entropy splitting lemma
(Lemma~\ref{lem:splitting}). Extraction is therefore inefficient, which makes
these results inapplicable in the context of zero-knowledge.

Nevertheless, it would be surprising if the verifier could learn anything from
the non-interactive proof that it could not learn in the $\Sigma$--protocol. The
security of BQS-OT ensures that the response to one of the two possible
challenges is hidden \emph{information theoretically}. The impossibility of
zero-knowledge appears to be more due to a lack of ways in which the simulator
can ``cheat'' than to an actual leakage of information. We can therefore show
that other security properties against malicious verifiers -- e.g.\ witness
hiding and witness indistinguishability -- are preserved by our transformation.

\subsubsection{Honest Verifier Zero-Knowledge.}
It is trivial to show that a simulator able to read the honest verifier's memory
after the memory bound is able to produce a valid proof. The simulator acts as
both the prover and the honest verifier: for each $i\in[k]$, it prepares the
states $\ket {x_i}_{\theta_i}$, picks a bit $c_i\in\bool$ at random and
measures the state in basis $c_i$. After the measurements with outcomes
$x'_1,\dots,x'_k$, the simulator uses the HVZK simulator for $\Sigma$ on input
$c_i$ to produce a valid transcript $(a_i,c_i,r^{c_i}_i)$. For the classical
prover message, the simulator chooses $h^0_i,h^1_i$ at random and sets
$m_i^{c_i}=r^{c_i}_i\oplus h^{c_i}_i(x^{c_i}_i)$ and $m_i^{1-c_i}$ uniformly
random. The simulator runs $V$ on the message
$(\theta_i,m_i^0,m_i^1,a_i,h_i^0,h_i^1)_{i\in [k]}$ outputs whatever $V$
outputs.

\subsubsection{Witness Indistinguishability.}

Witness indistinguishability was introduced
in~\cite{feige_witness_1990} as a relaxation of zero-knowledge. We adapt the
definition to quantum proof systems.

\begin{definition}[Witness Indistinguishability] Let $R$ be an \npol{} relation and
  let $\Pi$ be a quantum proof system for $R$. We say that $\Pi$ is
  computationally (resp.\ statistically) \emph{witness indistinguishable}
  (BQS-WI) if for any $V'$, for any instance $x$ and witnesses $w_1,w_2$, and
  any auxiliary input $y$, the quantum states
  \begin{equation*}
    \langle P(x,w_1), V'(x,y) \rangle \text{ and } \langle P(x,w_2), V'(x,y) \rangle
  \end{equation*}
  are computationally (resp.\ statistically) trace-indistinguishable. We say
  $\Pi$ is WI in the BQSM (BQS-WI) if indistinguishability holds for any
  $q$--bounded $V'$.
\end{definition}

\begin{theorem}
  If $\Pi$ is a (computational/statistical) witness indistinguishable proof
  system, then $\nizkot[\Pi]$ is (computational/statistical) witness
  indistinguishable in the BQSM.
\end{theorem}
\begin{proof}
  Let $w,w'$ be two witnesses for $x\in L$. Let $\rho_{TZE}$ be the state of the
  $q$--bounded verifier after interacting with $\prover(x,w)$ where $E$ is the
  $q$-qubit quantum memory of $\verifier$, $Z$ is its classical partial
  measurement outcome and
  $T=(\Theta^{(i)},H_0^{(i)},H_1^{(i)},M_0^{(i)},M_1^{(i)},A^{(i)})$ is the
  classical register sent by the prover. Let $\sigma_{TZE}$ be the state where
  the prover uses the witness $w'$ instead.

  By the security of BQS-OT (Theorem~\ref{thm:bqs-ot}), for each $i$ there
  exists a random variable $C_i$ such that $M^{(i)}_{1-C_i}$ is statistically
  close to independently and uniformly random. Let $c=c_1\dots c_k$ and let
  $\rho^c$ denote the state where $M^{(i)}_{1-c_i}$ is replaced
  with the completely mixed state for each $i$:
  \begin{equation}
    \rho^c= \frac 1{2^\ell}
    \id_{M_{\bar c}}\otimes \trace[M_{\bar c}]{\rho}
    \enspace.
  \end{equation}
  We define $\sigma^c$ in the same way. By Theorem~\ref{thm:bqs-ot}, $\rho
  \approx^{\epsilon} \sum_{c\in\bool^k}p_c \rho^c$ where $p_c=\Pr[C=c]$. By the
  witness indistinguishability of $\Pi$, we have that
  \begin{align*}
    \|\dist(\rho) - \dist(\sigma)\|
    &\leq \|\dist(\sum_cp_c \rho^c)  - \dist(\sum_c p_c \sigma^c)\| +2\epsilon \\
    &\leq \sum_c p_c\| \dist(\rho^c) - \dist(\sigma^c) \| +2\epsilon \\
    &\leq \nu + 2\epsilon
  \end{align*}
  since the probability to distinguish between $\rho^c$ and $\sigma^c$ is at
  most the distinguishing probability between a transcript for $\Pi$ with
  challenge $c$ and witness $w$ and one with witness $w'$.
  \qed
\end{proof}

\subsection{Non-interactive statistical WI proofs for NP }
\label{sec:bqsm-NIZK}

We describe now the application of our protocol for (statistical) WI
non-interactive proofs for NP. Before discussing such a protocol, we first
describe a new non-interactive weak bit-commitment, which may have
independent interest.

\subsubsection{A new non-interactive weak BC.}
\label{sec:new-bqs-bc}

The previous protocols for bit commitment in the BQSM had the weird property
that the sender commits by measuring a quantum state created by the
receiver.
For example, in \cite{doi:10.1137/060651343}, in order to commit to a message $m\in\bool$, the sender
would get a message $\ket x_\theta$, measure it in basis $m$ and take note of
the outcome $x'$. To open its commitment, it would send $m$ and $x'$ to the
receiver who could check that $x'_i=x_i$ whenever $\theta_i=m$. This quirk of
\dfss{} is actually the reason why our round compression transform can go down to
two messages as we will see in \Cref{sec:multi-round}. But in the context of our non-interactive proof using \dfss{}
applied to commit-and-open protocols, we cannot replace the classical
commitments with \dfss{} since it would introduce communication form the verifier
to the sender.

Intuitively, it does not matter who prepares the state and who measures it since
by a purification argument, the state preparation of \dfss{} can be seen as
measuring halves of EPR pairs. Formally proving that this is still secure is
more difficult, and the tools to do so were only discovered a couple of years
later in~\cite{dupuis_adaptive_2016}, which can show that it is still
sum-binding. For our purpose, we actually need a weaker security notion than
sum-binding. We first present the ``reversed'' protocol and then describe and
prove the security notion it needs to satisfy.

\begin{quote}
  \rule{\linewidth}{1pt}
  \begin{center}
    {\bf Protocol} weak-BC
  \end{center}

  {\bf Commit Phase }
  \begin{itemize}
  \item {\bf Committer$(b)$: } Choose $x\in_R\bool^n$. Send $\ket x_b$.
  \item  {\bf Receiver: } Measure qubits upon reception in a random basis $\theta$,
    gets outcome $x'$.
  \end{itemize}
  
  {\bf Open Phase }
  \begin{itemize}
  \item {\bf Committer$(b)$: } Send $x$ and $b$. Receiver checks that $x_i=x'_i$
    whenever $\theta_i=b$.
  \item {\bf Receiver: }  Check that $x_i=x'_i$ whenever $\theta_i=b$.
  \end{itemize}
  \rule{\linewidth}{1pt}
\end{quote}

The usual sum-binding criteria asks that, for a fixed commitment $\rho_{AB}$, if
the sender succeeds in opening $b$ with probability $p_b$, then $p_0+p_1\leq
1+\negl$. In this context, the malicious sender can measure its part of the
state $A$ adaptively based on the knowledge of the target bit $b$. We consider a
weaker task where the sender must provide \emph{both} openings simultaneously,
and does not know which will be tested. This is strictly weaker than sum-binding
since, as the following theorem shows, this is achieved unconditionally by the
above protocol.

\begin{theorem}\label{thm:weak-bc}
  The above weak-BC protocol is perfectly hiding and is binding according to the
  following. Let $\rho_{AB}$ be an arbitrary density operator describing the
  joint state of the committer and the receiver after the commit phase. Let
  $\{V_{acc}^{x,b},V_{rej}^{x,b}\}$ be the verifier's measurement for opening
  $(x,b)$. Then \begin{equation}
    \label{eq:sum-bind-no-q}
    \sup_{M}  \sum_{b\in\bool} \sum_{x_0,x_1}
    \trace{(M_{x_0,x_1}\otimes V_{acc}^{x_b,b}) \rho_{AB}}
    \leq 1+2^{-\frac n2 + 2h(\delta)n} + 2^{-\delta n+1}
  \end{equation}
  where $h(\cdot)$ is the binary entropy and $\delta>0$ is an arbitrary
  constant.
\end{theorem}
\begin{proof}
  Hiding follows from the fact that
  \begin{equation*}
    \sum_{x\in\bool^n} \proj x=
    \sum_{x\in\bool^n} H^{\otimes n} \proj x H^{\otimes} = \frac \id{2^n}
  \end{equation*}

  We will bound the weak binding criteria through a series of hybrids which each
  negligibly change the success probability. Let $p_b=\sup_M \sum_{x_0,x_1}
  \trace{(M_{x_0,x_1}\otimes V_{acc}^{x_b,b}) \rho_{AB}}$ be the probability of
  acceptance when the opening to bit $b$ is checked, where of course $M$ cannot
  depend on $b$. 

  Hybrid 1. The receiver holds on to the qubits in the commit phase and waits
  for the committer to send its opening before measuring in a random basis
  $\Theta$. The trace in~\eqref{eq:sum-bind-no-q} is unchanged by this
  modification.

  Hybrid 2. As Hybrid 1, but instead of choosing $\Theta$ at random and measuring in basis
  $\Theta$, the receiver measures all the qubits in the basis $b$ sent by the
  commiter. Then,  the receiver chooses a subset $T\subseteq [n]$ uniformly at random
  and rejects if for any $i\in
  T$, the result $x'_i$ is different from $x_i$. The probability distributions
  are also unchanged as this is equivalent to the checking procedure with
  $\Theta_i=b$ if $i\in T$ and $\Theta_i=1-b$ if $i\notin T$. The marginal
  distribution of $\Theta$ is still uniform.

  Hybrid 3. As Hybrid 2, but instead of comparing the positions for a random
  subset $T$, the receiver rejects if the measurement outcome $x'$ is
  at Hamming distance greater than $\delta n$ from $x$. The receiver will reject
  more often in this hybrid. The probability that the verifier rejects in Hybrid
  3 and not in Hybrid 2 is the probability that $\Delta(x',x)>\delta n$, yet
  $x'_i=x_i$ for all $i\in T$. Since $T$ is chosen uniformly at random, this
  probability is at most $2^{-\delta n}$. Let $p'_b$ be the probability that the
  receiver accepts an opening to $b$ in Hybrid 3, then $p_b\leq p'_b +
  2^{-\delta n}$.

  We now bound the sum of probabilities for Hybrid 3. Let $\sum_{x'\approx x}
  \proj{x'}_b$ be the projector onto accepting outcomes for the opening of $b\in
  \bool$ in Hybrid 3. We have that
  \begin{align*}
    p'_0 + p'_1 &= \sup_{M}  \sum_{b\in\bool} \sum_{x_0,x_1}
                  \trace{(M_{x_0,x_1}\otimes \sum_{x\approx x_b} \proj{x'}_b) \rho_{AB}}\\
                &\leq \sup_{x_0,x_1}\trace{\sum_{x\approx x_0} \proj{x}_0 \cdot \rho} + \trace{\sum_{y\approx x_1} \proj{y}_1\cdot \rho}\\
                &\leq \left\| \sum_{x\approx x_0} \proj{x}_0 + \sum_{y\approx x_1} \proj{y}_1 \right\|_\infty\\
                &\leq 1 +\left\| \sum_{x\approx x_0} \proj{x}_0 \cdot \sum_{y\approx x_1} \proj{y}_1 \right\|_\infty\\
                &\leq 1+ 2^{2h(\delta)n - n/2}
  \end{align*}
  where we use the inequality $\|A+B\|\leq 1+\|A\cdot B\|$ for projectors $A$
  and $B$ (a fact whose proof can be found in~\cite{bouman_all-but-one_2013}),
  the fact that there are at most $2^{h(\delta)n}$ strings at distance $\delta
  n$ from $x_b$ and that $\bra {x}_0\ket {y}_1=2^{-\frac n2}$ for any $x,y$.

  Compiling the error introduced with Hybrid 3, we have that
  \begin{equation*}
    \eqref{eq:sum-bind-no-q}=p_0+p_1\leq 1+  2^{2h(\delta)n - n/2}+ 2\cdot 2^{-\delta n}
  \end{equation*}
\qed
\end{proof}

\subsubsection{A non-interactive statistical WI proof for \npol{} in the BQSM.}
\label{sec:ni-np}

We now consider the following $\Xi$ protocol for the NP-complete $L_{ham}$
corresponding to graphs that have a Hamiltonian cycle. It consists of the
original $\Sigma$ protocol for this problem, but using weak BC as the
commitment.

\begin{quote}
  \rule{\linewidth}{1pt}
  \begin{center}
    {\bf Protocol} $\Xi$ protocol $\Pi_{ham}$ for Hamiltonian cycle
  \end{center}

  \begin{enumerate}
  \item {\bf Prover: } Using weak BC, commits to the adjacency matrix of a
    permutation $\sigma$ of the graph $G$
  \item  {\bf Verifier: } Send a random bit $b$
  \item {\bf Prover: } If $b = 0$ open the whole adjacency matrix and provide
    the permutation $\sigma$. If $b = 1$, open the edges corresponding to the
      Hamiltonian cycle.
  \item  {\bf Verifier: } Check the consistency of the Prover's opening.
  \end{enumerate}
  \rule{\linewidth}{1pt}
\end{quote}

The completeness of the protocol follows directly from the completeness of the
original protocol, and zero-knowledge follows Watrous rewinding technique.
However, since we use weak BC, this protocol does not satisfy the standard
soundness definition (in particular, the Prover can answer the two challenges by
keeping the purification of the commitment and measuring it accordingly).

However, we prove now that it satisfies the {\em oblivious soundness} property that
we mentioned in \Cref{remark:oblivious-soundness}.

\begin{lemma}\label{lem:pi-ham-soundness}
  $\Pi_{ham}$ satisfies oblivious soundness.
\end{lemma}
\begin{proof}
  Let $G\notin L_{ham}$.
  Let $G'$ be the graph corresponding to the answer $r_0$. If $G'$ has a
  Hamiltonian cycle, it cannot be a permutation of $G$, therefore the first
  check will fail with probability $1$. Moreover, if $r_1$ does not open to a
  Hamiltonian cycle, the second check will fail with probability $1$. In this
  case, for the two checks to pass, there is at least one entry
  $i,j$ of the adjacency matrix whose opening $o_{i,j}$ is $b$ in $r_0$ and whose
  opening $o'_{i,j}$ is $\neg b$ in $r_1$. 

  Therefore, in order to make the verifier accept, the prover has to provide values $(r_0,r_1)$ such that
  \Cref{eq:oblivious-soundness} holds, which is upper-bounded by the probability
  that the prover can provide simultaneously two different openings to the
  commitment, which is impossible by \Cref{thm:weak-bc}. \qed
\end{proof}

By observing that $\Pi_{ham}$ is perfectly witness indistinguishable because
weak-BC is perfectly hiding, and combining Lemma~\ref{lem:pi-ham-soundness} and
Remark~\ref{remark:oblivious-soundness}, we obtain the following result. 

\begin{corollary}
  There is a non-interactive quantum proof system for $L_{ham}$ which is
  unconditionally sound and witness indistinguishable against BQS verifiers.
\end{corollary}

\section{A General Round-Compression Transform in the \bqsm}
\label{sec:multi-round}
In this section, we present and prove the soundness of the general transform
mapping $k$--round interactive proofs for $k=\poly[\lambda]$ to
$2$--message quantum proofs.

We assume for simplicity that all the prover messages are of length
$\ell=\ell(\lambda)$ and all the verifier challenges are of length $m=m(\lambda)$ for
some polynomials $n,m:\naturals\rightarrow\naturals$ , and that the prover sends
the first and last messages. We let $a_1,\dots, a_{k+1}$ denote the $k+1$ prover
messages and $c_1,\dots,c_k$ the $k$ verifier challenges, where $a_{i+1}$
responds to challenge $c_i$. Let $\prover^{\Pi}_i$ denote the next-message
function of the prover in protocol $\Pi$ that takes as input the partial
transcript so far and outputs $a_i$. The $\pfs$ transform is presented below.
\begin{quote}
  \rule{\linewidth}{1pt}
  \begin{center}
    {\bf Protocol} $\pfs[\Pi]$
  \end{center}

  {\bf Parameter: } A $k$--round interactive proof system
  $\Pi=(\prover^\Pi,\verifier^\Pi)$ for a language $L$.

  \vspace{1em}

  {\bf Verifier message: } 
  \begin{enumerate}
  \item 
    For $i\in [k]$:
    \begin{enumerate}
    \item \verifier{} runs the commit phase of the \dfss{} string commitment to
      get a quantum register $P_i$.
\item \verifier{} picks $c_i\in_R \bool^m$ to initialize a 
      register $C_i$ in state $\ket{c_i}$
    \end{enumerate}
  \item \verifier{} sends the registers $P_1C_1\dots P_kC_k$ to \prover.
  \end{enumerate}

  \begin{center}
    \emph{--- A memory bound applies after transmission of each $P_i$ ---}
  \end{center}

  {\bf Prover message: }

  \begin{enumerate}
  \setcounter{enumi}{2}
  
\item On input $x\in L$, \prover{} first computes $a_1=\prover^{\Pi}_1(x)$.
\item  For $i\in [k]$,
  \begin{enumerate}
  \item On reception of register $P_i$, \prover{} commits to $a_i$ as in the
    commit phase of \dfss{}. 
\item \prover{} measures register $C_i$ in the computational basis to obtain
    $c_i$. \prover{} computes $a_{i+1} =
    \prover^{\Pi}_{i+1}(a_1,\dots,a_{i},c_1,\dots,c_{i},x)$.
  \end{enumerate}
\item \prover{} runs the reveal phase of \dfss{}, sending every $a_i$ and opening string to
  \verifier{}.
  \end{enumerate}

  \vspace{1em}

  {\bf Verification: } 
  \begin{enumerate}
  \setcounter{enumi}{5}
\item \verifier\ performs the verification for every instance of \dfss{}. It
  accepts if every opening is valid and if $a_1,\dots,a_{k+1},c_1,\dots,c_k$ is
  an accepting transcript for $\Pi$ on input $x$. Otherwise, it rejects.
\end{enumerate}

  \rule{\linewidth}{1pt}
\end{quote}

\begin{theorem}\label{thm:soundness-main}
Let \dfss{} be the $\delta$--binding BQS-BC from Section~\ref{sec:bc-bqsm}. If
  $\Pi$ is a $k$--round interactive proof with soundness error
  $\epsilon$ against unbounded (resp. \qpt{}) provers, then $\pfs[\Pi]$ is a
  1--round quantum interactive proof (resp. argument) with soundness error
  \begin{equation}
    \label{eq:soundness-error}
    \epsilon+ k^2\cdot \delta \end{equation}
  against $q$--bounded adversaries where $\delta$ is negligible if
  $n/4-q\in\Omega(\lambda)$ where $n=n(\lambda)$ is the number of qubits sent in
  \dfss{} and $q=q(\lambda)$ is the quantum memory bound on the prover.
\end{theorem}
\begin{proof}

  We use a hybrid argument to prove the soundness of $\pfs[\Pi]$. Consider
  the following hybrid protocols where in $\hyb\ i$ the round-compression
  transform is applied up to the $i$th prover message, and the rest of protocol
  is interactive.
  \begin{itemize}
  \item $\hyb\ 0$: same as protocol $\Pi$
  \item $\hyb\ i$: apply transformation $\pfs$ to the messages of $\Pi$ up to
    round $i$.
    \begin{enumerate}
    \item\label{st:hyb-after-msg} \verifier\ prepares $i$ registers $P_1\dots P_i$
      and $i$ random values $c_1,\dots,c_i$ in registers $C_1\dots C_i$ and
      sends $\bigotimes_{j=1}^i P_jC_j$.
    \item On reception of a message $(a_1\dots a_{i+1},z_1\dots z_i)$ from the
      prover, \verifier\ checks that $(a_j,z_j)$ is valid opening for $j\in [i]$
      and rejects if any are invalid. 
    \item \verifier\ and \prover\ continue as in protocol $\Pi$: \verifier\
      sends $c_{j}$ and \prover\ responds with $a_{j+1}$ for $j=i+1,\dots,k$.
      \verifier\ checks that $(a_1\dots a_{k+1},c_1\dots c_k)$ is an accepting
      transcript for $\Pi$.
    \end{enumerate}
  \item $\hyb\ k$: same as in $\pfs[\Pi]$
  \end{itemize}

  The difference between two hybrids $i-1$ and $i$ is that in hybrid
  $i-1$, $A_1\dots A_i$ are sent to $\verifier$ before it sends $C_i$ whereas in hybrid $i$,
  the adversary receives $C_i$ before opening its commitments to $A_1\dots A_i$.
We will show that this only confers a negligible advantage to an adversary.

  Consider a $q$--bounded adversary $\adv^i$ against $\hyb\ i$. By the
  definition of binding for BQS-BC (Definition~\ref{def:binding-DFRSS}), after
  the commit phase of the $j$th commitment (i.e.\ after the transmission of
  register $P_j$ for $j\leq i$), there is a random variable $A'_j$ such that
  conditioned on $A'_j=a'_j$, $\adv^i$ has negligible probability of opening
  the $j$th commitment to $a_j\neq a'_j$. This random variable is defined by the
  partial measurement $\adv^i$ is forced to make on register $P_j$ before
  \verifier{} begins transmission of register $C_j$, so it is independent of
  $C_j$.

  This independence means that learning $C_i$ before sending $A_1\dots A_i$ does
  not give a noticeable advantage to the adversary. We make this formal by
  constructing, from the adversary $\adv^i$ that has success probability
  $\epsilon_i$ against $\hyb\ i$, an adversary $\adv^{i-1}$ against hybrid $i-1$
  that has success probability at least $\epsilon_i-\negl[\lambda]$.
  $\adv^{i-1}$ performs the same strategy as $\adv^i$ on reception of the
  registers $P_1C_1\dots P_{i-1}C_{i-1}$. For producing the next value $a_i$,
  $\adv^{i-1}$ simulates the verifier in the $i$th commitment, i.e.\ creates the
  register $P_i$ just as \verifier{} would in hybrid $i$, again applying
  $\adv^i$'s strategy, and checking that the resulting opening is
  valid. 

\begin{quote}
  \rule{\linewidth}{1pt}
  \begin{center}
    {\bf Adversary}  $\adv^{i-1}$
  \end{center}

  \begin{enumerate}
  \item While receiving registers $\bigotimes_{j=1}^{i-1}P_jC_j$ from the
    verifier, forward them to $\adv^{i}$. For the last registers $P_{i}C_{i}$
    that $\adv^{i}$ expects, $\adv^{i-1}$ simulates the verifier, i.e.
    constructs register $P_i$ from the commit phase and sends $P_i$
followed by a random challenge $c$ to
    $\adv^{i}$.
  \item $\adv^{i-1}$ now receives a message $(a_1\dots a_{i+1}, z_1\dots z_{i})$
    from $\adv^{i}$. It checks $(a_{i},z_{i})$ is a valid opening of the
    commitment and aborts if the check fails. It discards $a_{i+1}$ and sends
    $(a_1\dots a_{i},z_1\dots z_{i-1})$ to the verifier.
  \item After receiving the challenge $c_{j}$ for $i\leq j\leq k$ from the
    verifier, it computes and sends $a_{j+1}$ using the same strategy as
    $\adv^i$.
  \end{enumerate}
  
  \rule{\linewidth}{1pt}
\end{quote}
Observe the following facts about $\adv^{i-1}$:
\begin{itemize}
\item The quantum memory required to perform attack $\adv^{i-1}$ against $\hyb\ i-1$
  is the same as attack $\adv^{i}$ against $\hyb\ i$.
\item  $\adv^{i}$ cannot distinguish whether it is interacting
  with $\verifier$ in $\hyb\ i$ or with $\adv^{i-1}$ in $\hyb\ i-1$.\item The random variables $A_1',\dots,A_i'$ have the same distribution in both
  experiments ($\adv^i$ against $\hyb\ i$ and $\adv^{i-1}$ against $\hyb\ i-1$).
\end{itemize}

Let us fix somesome  arbitrary values $a_1'\dots a_i'$ for $A_1'\dots A_i'$. Assume for
now that $a_1\dots a_i=a'_1\dots a'_i$. Since these values are independent of
$C_i$, they would remain unchanged for any value $c_i$ that $\adv^{i-1}$ had
given to $\adv^i$. 
And since $\adv^{i-1}$ answers the
rest of the challenges exactly as $\adv^i$ would, the whole transcript is
identically distributed in both experiments, thus the probability of the
verifier accepting is the same.

Now for the other case (there is some $j\leq i$ such that $a_j\neq a'_j$), then
the verifier will reject the opening to the $i$th commitment with overwhelming
probability. By the definition of $\delta$--binding for BQS-BC schemes
(Definition~\ref{def:binding-DFRSS}), the probability that $\adv^{i+1}$ can
announce a basis $A_j\neq A'_j$ is upper-bounded by
$\delta$. By a union bound, the probability that there is a $1\leq j\leq i$ such that
$A_j\neq A'_j$ is at most $i\cdot
\delta$. Therefore if we let $\epsilon_j$ denote the soundness error of $\hyb\ j$ for
$j=0\dots k$, then
\begin{equation}
  \epsilon_{i} \leq \epsilon_{i-1}+ i\cdot\delta\enspace. \end{equation}
Since by assumption, $\Pi$ is $\epsilon$--sound, then
$\pfs[\Pi]$ is $\epsilon'$--sound for
\begin{equation}
  \epsilon'\leq \epsilon+ k^2\cdot \delta \end{equation}
where the $k^2$ comes from the fact that when going from hybrid $i$ to hybrid
$i+1$, we introduce the negligible term $i\leq k$ times, and there are $k$
hybrids. If $\delta$ is negligible, then the above is arbitrarily close to
$\epsilon$. By Theorem~\ref{thm:binding-dfss}, this happens if the memory bound
on the prover satisfies $n/4-q\in\Omega(\lambda)$

\qed

\end{proof}

\subsection{Application: Two-Message Zero-Knowledge in the BQSM}
\label{sec:zk-bqsm}

The goal of this section is to construct a two-message zero-knowledge proof for
any $\npol$ language in the BQSM.
We begin by proving that our transform produces a zero-knowledge $2$--message
quantum proof when applied to proof systems that satisfy the
following notion of honest-verifier zero-knowledge, which is a generalization of
special HVZK to multi-round protocols.

\begin{definition}
  We say that a $\Pi$ protocol is special qHVCZK (special HVSZK) if for any given challenge
  $(c_1,...,c_k)$, there is an efficient simulator $\mathcal{S}(c_1,...,c_k)$
  such that for every \qpt{} (unbounded) distinguisher $\mathcal{D}$, 
  \begin{align*}
    |\Pr[\mathcal{D}^{\mathcal{S}(x,\cdot)}(1^\lambda) = 1] - 
    \Pr[\mathcal{D}^{\prover \leftrightharpoons \verifier(x,\cdot)}(1^\lambda
    ) = 1]| \leq
    \negl[\lambda],
  \end{align*}
  where $\mathcal{D}$ can query its oracle with (classical) values
  $c_1,...,c_k$. In the first term, it receives $\mathcal{S}(x,c_1,...,c_k)$ and
  in the second term, it receives the transcript $\prover \leftrightharpoons
  \verifier(x,c_1,\dots,c_k)$ that come from the real protocol when the
  challenges are fixed.
\end{definition}

We now show that if $\pfs$ is applied to a special-HVZK $k$--round protocol for
$k=\poly[\lambda]$, then the resulting scheme is zero-knowledge against quantum
verifiers. To prove zero-knowledge, instead of producing a simulator for the
malicious verifier, we show that there exists a simulator which does not
interact with the prover and that can simulate the actions of $\prover$. One can
easily see that this implies (auxiliary-input) zero-knowledge by running any
malicious verifier $\tilde \verifier$ with this simulated prover.

At first glance, the existence of this simulator appears to be at odds with the
soundness of our transform. For example, if the prover relies on the knowledge
of a witness $w$ that $x\in L$ for $L\in\npol$, then the simulator can convince
the verifier that $x\in L$ without access to $w$. This matter is resolved by
observing that the quantum memory of the simulator is not bounded, unlike the
(malicious) prover. This fact is crucial, as we show in
Section~\ref{sec:impo-aux} that there are no $2$--message quantum proof systems
for hard languages that are both sound and zero-knowledge when the quantum
memory of the prover is not bounded. Furthermore, the existence of a fully
quantum simulator for a BQS adversary appears vacuous, but the party we are
simulating -- the verifier -- is not quantum memory bounded. Thus zero-knowledge
holds against fully quantum verifiers.

\begin{theorem}\label{thm:zk}
  If $\Pi=(\prover^\Pi,\verifier^\Pi)$ is a special qHVZK $\Sigma$--protocol
  for a language $L$, then
  $\pfs[\Pi]$ is qZK. The type of zero-knowledge (computational or
  statistical) is preserved by $\pfs$.
\end{theorem}
\begin{proof}
  We construct a simulator for the prover instead of the verifier; i.e.\ this
  simulator mimics the actions of the prover from the verifier's point of view
  and does not have access to the real prover. Turning this simulator into one
  for the verifier is then just a question of making the verifier interact with
  this simulated prover.

  First observe that from the verifier's point of view, the action of the
  quantum memory-less honest prover $\prover$ is perfectly indistinguishable
  from the action of a ``semi-honest'' prover $\prover^*$ that \emph{does} have
  a quantum memory and that delays its commitment to $a_i$ using $P_i$ until
  after every challenge $c_i$ was measured.

  Now since the prover messages $a_i$ are committed to after every challenge is
  learned, we can employ the simulator $\simulator_\Pi$ for the $\Sigma$-protocol
  to obtain a simulated transcript $(a_1,\dots,a_{k+1})$ indistinguishable from
  a real transcript. In more details, we construct the simulator $\simulator$
  for $\pfs[\Pi]$ as follows:
  \begin{enumerate}
  \item Receive the registers $P_1,C_1,...,P_k,C_k$ from $\tilde{\verifier}$, delaying
    any measurement
    \item Measure registers $C_1,...,C_k$ in the computational basis and get
      outcomes $c_1,...,c_k$
    \item Compute $\simulator_\Pi(c_1,...,c_k) = (a_1,...,a_{k+1})$, where
      $\simulator_\Pi$ is the special qHVZK simulator
    \item Perform the commitment phase of BQS-BC on register $P_i$ by committing
      to $a_i$ and get the opening string $z_i$
    \item Return $(a_1,...,a_{k+1},z_1,...,z_k)$ to $\tilde{\verifier}$
  \end{enumerate}

  We now show that this simulator indistinguishable from $\prover$. For that,
  let us assume towards a contradiction that there exists a distinguisher
  $\mathcal{D}$ and a state $\rho_{QE}$, where $Q=P_1C_1\dots C_kP_{k}$ is sent to the
  prover/simulator and $E$ is kept by the distinguisher, such that
  \begin{equation}
    \|\mathcal{D}(\prover\otimes \id_E(\rho)) - \mathcal{D}(\simulator\otimes \id_E(\rho))\| \geq \lambda^{-d}
  \end{equation}
  for $d\in O(1)$. Then, we can construct a distinguisher
  $\mathcal{D}^\mathcal{C}_\Pi$ that can break the special qHVZK property of
  $\Pi$ with probability at least $ \lambda^{-d}$, where $\mathcal{C}$ is an
  oracle for either $\prover_{\Pi} \leftrightharpoons \verifier_{\Pi}(x,\cdot)$ or
  $\simulator_\Pi(x,\cdot)$. It works as follows:
  \begin{enumerate}
  \item Compute the state $\rho_{QE}$ which allows to distinguish $\simulator$
    and $\prover$
  \item Measure registers $C_1,\dots,C_k$ of $\rho_{Q}$ and get outcome
    $(c_1,...,c_k)$
  \item Query $\mathcal{C}(c_1,...,c_k)$ and get the output
    $(a_1,...,a_{k+1})$
  \item Commit to $a_i$ using register $P_i$ and get opening string $z_i$
  \item Output $\mathcal{D}(a_1,...,a_{k+1},z_1,...,z_k)$
  \end{enumerate}

  Notice that when $\mathcal{C} = \simulator_\Pi(x,\cdot)$, then the output
  of $\mathcal{D}^\mathcal{C}_\Pi$ is $\mathcal{D}(\simulator\otimes
  \id_E(\rho))$. Moreover, when $\mathcal{C} = \prover_\Pi \leftrightharpoons
  \verifier_\Pi(x,\cdot)$, we have that $\mathcal{D}^\mathcal{C}_\Pi$ is
  $\mathcal{D}(\prover^* \otimes \id_E(\rho))$ where $\prover^*$ is the
  semi-honest prover introduced earlier. In this case, we have that
\begin{equation*}
  \|\mathcal{D}^{\prover_\Pi \leftrightharpoons
    \verifier_\Pi(x,\cdot)}_\Pi(1^\lambda) - \mathcal{D}^{\simulator_\Pi(x,\cdot)}_\Pi(1^\lambda) \| =
  \|\mathcal{D}(\prover^*\otimes \id_E(\rho)) - \mathcal{D}(\simulator\otimes
  \id_E(\rho))\| \geq \lambda^{-d},
\end{equation*}
which contradicts the qHVZK of $\Pi$ by recalling that the actions of
$\prover^*$ and $\prover$ are perfectly indistinguishable. Therefore we conclude
that the CPTP maps $\prover$ and $\simulator$ are (computationally or
statistically) indistinguishable if $\Pi$ is (computationally or
statistically) qHVZK. \qed
\end{proof}

\subsubsection{Quantum statistical zero-knowledge proofs.}
\label{sec:qszk-proofs}

In this section, we show that using the statistically binding and hiding BQS-BC
scheme of Section~\ref{sec:bc-bqsm}, we can achieve 2--message quantum
statistical zero-knowledge proofs in the \bqsm{}.

In the previous subsection, we showed that special qHVZK $\Sigma$ protocols can
be converted into $2$-messages QZK protocols in the \bqsm{}. However, (honest
verifier) ZK proofs for \npol{}-complete languages rely on computational
assumptions, usually to implement commitment schemes. Since we are in \bqsm{},
we can instead use quantum commitment schemes with perfect hiding and
statistical binding and achieve statistical ZK proofs in the \bqsm{}. 

For simplicity, we will prove the result for a single-shot run of $3$-coloring,
but the result follows analogously with the parallel repetition of the protocol.

\begin{quote}
  \rule{\linewidth}{1pt}
  \begin{center}
$2$--message perfect zero-knowledge proof
  \end{center}

  {\bf Input: } Graph $G=(V,E)$ with $|V| = n$.

  {\bf Verifier message: } 
  \begin{enumerate}
  \item For $i = 1,...,n$, \verifier{} runs the commit phase of the \dfss{}
    string commitment to get a quantum register $P_i$.
  \item \verifier{} picks $c\in_R E$ to initialize a register $C$ in state
    $\ket{c}$
  \item \verifier{} sends the registers $P_1\dots P_n C$ to \prover.
  \end{enumerate}

  \begin{center}
    \emph{--- A memory bound applies after transmission of each $P_i$ ---}
  \end{center}

  {\bf Prover message: }

  \begin{enumerate}
  \setcounter{enumi}{2}
  
\item \prover{} first computes a random $3$-coloring of the graph $G$. Let
  $w_1$,...,$w_n$ be the color of each vertex of the graph.  \prover{} commits
  to each of the colors independently: for $i\in[n]$,
\item On reception of register $P_i$, \prover{} commits to $w_i$ as in the
  commit phase of \dfss{}.
\item \prover{} measures register $C$ in the computational basis to obtain
  $\{i,j\} \in E$.
\item \prover{} runs the reveal phase of $\dfss$ for $w_i$ and $w_j$.
  \end{enumerate}

  \vspace{1em}

  {\bf Verification: } 
  \begin{enumerate}
  \setcounter{enumi}{5}
\item \verifier{} runs the verification of \dfss{} for $w_j$ and $w_i$ and
  checks that $w_j \ne w_i$.
\item If verification or the check failed, it aborts. Otherwise, it accepts.
\end{enumerate}

  \rule{\linewidth}{1pt}
\end{quote}

\begin{theorem}
The protocol described above is a two-message perfect zero-knowledge proof for
  $3$-coloring.
\end{theorem}
\begin{proof}
  Completeness follows straightforwardly if the $\prover$ follows the honest
  strategy.

  To prove soundness, we use the $\epsilon$-binding property of the commitment
  scheme. For that, let $w'_1,...,w_n'$ be values of the the random variables $b'_1,...,b_n'$ that come from
  \Cref{def:binding-DFRSS} corresponding to the commitment of the color of each
  node. We notice that since the graph is not $3$-colorable, there exists at
  least one edge $\{i,j\} \in E$ such that $w'_i = w'_j$. We also have that the
  $\verifier$'s challenge is $\{i,j\}$ with probability $\frac{1}{m}$, and let
  us consider this case.

  If $\prover$ opens the commitments to the values $w'_i$ and $w'_j$,
  $\verifier$ rejects with probability $1$. If $\prover$ opens the commitments
  to values $\tilde{w}_i \ne w'_i$ or $\tilde{w}_j \ne w'_j$, $\verifier$
  rejects except with probability $\epsilon$. 

  In this case, if the graph is not $3$-colorable, $\verifier$ rejects with
  probability at least $\frac{1-\epsilon}{m}$.

  The simulator and the zero-knowledge proofs follow closely the proof of
  \Cref{thm:zk}. The fact that the flavour of zero-knowledge is perfect comes
  from the fact that the commitment scheme has perfect hiding since no
  information of non-open values is sent to $\verifier$.
  \qed
\end{proof}

\subsection{Applications: Two-Message Interactive Proof for \pspace{}}
\label{sec:applications}
In this section, we describe applications of our round compression transform
$\pfs$ presented in~\ref{sec:multi-round} when applied to a specific interactive
proof system.

\subsubsection{Sum-Check Protocol.}
\label{sec:sum-check-protocol}
The sum-check protocol is the key ingredient of several fundamental results in
complexity theory and cryptography. In this protocol, the prover aims to prove
that 
\[ \sum_{x_1,\dots,x_n \in \{0,1\}} f(x_1,\dots,x_n) = B,\]
for some given value $B$ and function $f$ an $n$-variate polynomial of
degree at most $d$.
The idea of the sum-check protocol is to consider a field 
$\mathbb{H}$, where $\mathbb{F}_2 \subseteq \mathbb{H}$ and $|\mathbb{H}| \gg d$,

\begin{quote}
  \rule{\linewidth}{1pt}
  \begin{center}
    {\bf Sum-check Protocol}
  \end{center}

  {\bf Prover $1$st message: } $\prover$ computes $g_1(x_1) = \sum_{x_2,\dots,x_n \in
  \{0,1\}} f(x_1,\dots,x_n)$ and
  sends $g_1$ to $\verifier$, who checks that $g_1$ is an univariate polynomial of
  degree at most $d$ and that $g_1(0) + g_1(1) = B$. If any of the checks failed, reject.

  \vspace{1em}

  {\bf Verifier $1$st message: } $\verifier$ sends a uniformly random $r_1 \in
  \mathbb{H}$ to $\prover$.

  \vspace{1em}

  {\bf Prover $i$th  message: } $\prover$ computes
  \begin{equation*}
    g_i(x_i) = \sum_{x_{i+1},\dots,x_n \in
      \{0,1\}} f(r_1,\dots,r_{i-1},x_i,x_{i+1},\dots,x_n)
  \end{equation*}
and
  sends $g_i$ to $\verifier$, 
  who checks that $g_i$ is an univariate polynomial of
  degree at most $d$ and that $g_i(0) + g_i(1) = g_{i-1}(r_{i-1})$. If any of the checks failed, reject.

  \vspace{1em}

  {\bf Verifier $i$th message: } $\verifier$ sends a uniformly random $r_i \in \mathbb{H}$ to $\prover$.

  \vspace{1em}

  {\bf Prover last message: } $\prover$ computes $g_n(x_n) =
  f(r_1,\dots,r_{n-1},x_n)$ and
  sends $g_n$ to $\verifier$,
  who checks that $g_n$ is an univariate polynomial of
  degree at most $d$ and that $g_n(0) + g_n(1) = g_{n-1}(r_{n-1})$.  Moreover, $\verifier$ also checks that $g_n(r_1,\dots,r_n) = f(r_1,\dots,r_n)$, for a random
  $r_n \in \mathbb{H}$. If either of these tests do no pass,
  reject.

  If all tests passed, $\verifier$ accepts.

  \rule{\linewidth}{1pt}
\end{quote}

The main result regarding the sum-check protocol is the following~\cite{LundFKN92,Shamir}.
\begin{theorem}
  The sum-check protocol presented above has the following properties:
  \begin{itemize}
    \item {\bf Completeness:}
If $\sum_{x_1,\dots,x_n \in \{0,1\}} f(x_1,\dots,x_n) = B$, there is a strategy for
$\prover$ such that 
      $\verifier$ accepts with probability  $1$.
    \item {\bf Soundness:} If $\sum_{x_1,\dots,x_n \in \{0,1\}} f(x_1,\dots,x_n) \ne B$,
      for any strategy for $\prover$, $\verifier$ accepts with probability at most $\frac{nd}{|\mathbb{H}|}$.
    \item {\bf Complexity:} The honest prover runs in time $\poly[|\mathbb{H}|^n]$, the
      verifier runs in time $\poly[|\mathbb{H}|,n]$ and space
      $O(n\log{|\mathbb{H}|})$. The communication complexity is
      $\poly[|\mathbb{F}|,n]$ and the number of bits sent by the verifier is
      $O(m \log{|\mathbb{H}|})$. Moreover, the protocol is public-coin.
  \end{itemize}
\end{theorem}

We notice that the sum-check protocol is a multi-round interactive proof where
$\verifier$ only sends random coins (interpreted as field elements) as messages.
In this case, we can apply the \pfs{} transformation to it to achieve a one-round quantum
protocol with similar guarantees of the classical sum-check protocol in the
quantum bounded storage model. 

\begin{corollary}[2--Message Quantum Sum-Check Protocol]
  \label{thm:quantum-sum-check}
  There is a $2$--message quantum proof for the sumcheck problem in the bounded
  storage model with negligible soundness error against $q$--bounded provers for
  an appropriate bound $q$.
\end{corollary}

The sum-check protocol is a crucial tool in results in complexity theory and
cryptography, especially regarding delegation of (classical) computation. 
We can easily replace the classical sum-check protocol by its quantum version to
to achieve round-efficient protocols, that we describe below.

\begin{corollary}
  Every language in \pspace{} has a $2$--message quantum protocol in the bounded storage model.
\end{corollary}

Notice that if we do not consider provers with bounded memory, we have that $\pspace{} = \qip(3)$, and if we define $\qiptwo^{\mathsf{BQSM}}$ as the class of problems with two--message quantum interactive proof systems where the prover has bounded quantum memory (but unbounded computational power), we have that   
$\pspace{} = \qiptwo^{\mathsf{BQSM}}$.

More recently, the sum-check protocol has been also used to achieve protocols
for doubly-efficient delegation of computation. In this setting, the goal is to
achieve a protocol where $\verifier$ interacts with a $\prover$ in order to
delegate the computation of an arithmetic circuit with the following properties:
\begin{itemize}
  \item The honest prover's computation should not be much more costly than
    running the original circuit.
  \item The running time of the verifier should be linear in the input size of
    the circuit.
\end{itemize}

Such a protocol was originally proposed by Goldwasser, Kalai and
Rothblum~\cite{GKR15} and later improved by Reingold, Rothblum and
Rothblum~\cite{RRR21}

\begin{lemma}[Corollary 1.4 of \cite{GKR15}]
Let $L$ be a language in \pol{}, that is, one that can be computed by a
deterministic Turing machine in time $poly(n)$. There is an interactive proof
for $L$ where:
\begin{itemize}
\item the honest prover runs in time $\poly$ and the verifier in time
  $\poly$ and space $O(log(n))$;
\item the protocol has perfect completeness and soundness $1/2$; and
\item the protocol is public coin, with communication complexity $\poly$.
\end{itemize}
\end{lemma}

Again here, the interaction between the verifier and the prover consists of
multiple instances of the classical sum-check protocol. Therefore, using
\Cref{thm:quantum-sum-check}, we achieve the following.

\begin{corollary}
Let $L$ be a language in \pol{}. $L$ has a quantum interactive proof in the bounded
  storage model where:
  \begin{itemize}
 \item the honest prover runs in time $\poly$ and the verifier in time
    $\poly$ and space $O(log(n))$;
  \item the protocol has perfect completeness and soundness $1/2$; and
  \item there is one round of communication.
  \end{itemize}
\end{corollary}

\section*{Acknowledgements}
\label{sec:acks}

We are grateful to Serge Fehr for enlightening discussions on the all-but-one
uncertainty relation (Theorem~\ref{thm:all-but-one}) and to Louis Salvail for
discussions on simulating bounded storage adversaries.
ABG is supported by ANR JCJC TCS-NISQ ANR-22-CE47-0004 and by the PEPR integrated
project EPiQ ANR-22-PETQ-0007 part of Plan France 2030.
This work was done in part while ABG was visiting the Simons Institute for the Theory of Computing.
\printbibliography

\appendix

\section{A New String Commitment Scheme in the \bqsm{}}
\label{sec:new-bqsm-bc}

The starting point of our new protocol is a more powerful uncertainty relation
found in~\cite{bouman_all-but-one_2013} and described below. We present our
protocol, which we call $\ourbc$ for the all-but-one uncertainty relation it
crucially relies on, and prove its security.

\subsection{All-but-one Uncertainty Relation}
\label{sec:all-but-one}

We use an uncertainty relation from~\cite{bouman_all-but-one_2013}. It
states that for a given quantum state $\rho$ and a family of bases
$\mathcal{B}_1,\dots,\mathcal{B}_n$ that have a small overlap, there exists a
basis $J'$ (defined as a random variable whose distribution depends on $\rho$)
such that for any other basis $J\neq J'$, the uncertainty of the measurement
outcome in basis $J$ is high.

Formally, let $\mathcal{B}_i:= \{\ket x_i \mid x\in\bool^N\}$ and define the
maximal overlap of bases $\mathcal{B}_1,\dots,\mathcal{B}_n$ as $c:= \max\{ \bra
x_i\ket y_j\mid x,y\in\bool^N, i\neq j \}$. Let $\delta:= -\frac 1n \log c^2$.
The uncertainty relation is as follows.

\begin{theorem}[Theorem 9 of~\cite{bouman_all-but-one_2013}]\label{thm:all-but-one}
  Let $\rho$ be an arbitrary $N$--qubit state, let $J$ be a random variable over
  $[n]$ with distribution $P_J$, and let $X$ be the outcome of measuring $\rho$
  in basis $\mathcal{B}_J$. Then for any $0<\epsilon<\delta/4$, there exists a
  random variable $J'$ with joint distribution $P_{JJ'X}$ such that
  \begin{itemize}
  \item $J$ and $J'$ are independent and
  \item there exists an event $\Psi$ with $\Pr[\Psi]\geq 1-2\cdot 2^{-\epsilon
      n}$ such that 
  \end{itemize}
  \begin{equation}
    \label{eq:all-but-one}
    H_\infty(X| J=j,J'=j',\Psi) \geq \left( \frac \delta2 -2\epsilon \right)N-1
  \end{equation}
  for all $j,j'\in[n]$ with $j\neq j$ and $P_{JJ'|\Psi}(j,j')>0$.
\end{theorem}
As emphasized in~\cite{bouman_all-but-one_2013}, the distribution of $J$ does
not need to be set for $J'$ to be well defined. In particular, the distribution
of $J'$ is fully determined by $\rho$. 

We now present how to efficiently construct a family of bases with large overlap
$\delta$. Let $G$ be the generator matrix of a linear $[N,n,d]$--error
correcting code. Then for the family of bases defined by
\begin{equation}
  \label{eq:defBj}
  \mathcal{B}_j:=\{(H^{c_1}\otimes\dots\otimes H^{c_N}) \ket x\mid x\in\bool^N,
  c=G\cdot j\}
\end{equation}
for $j\in\bool^n$ satisfies $\delta=\frac dN$.

\subsection{The Commitment Scheme}
\label{sec:abo-bc}

Our new bit commitment scheme is presented below. The intuition behind the
scheme is that the basis used by the committer to commit to a string $a$ should
be far from the basis of $a'\neq a$. Therefore, we can use code words of an
error correcting code as the bases to ensure this distance holds. The original
\dfss{} scheme (presented in Section~\ref{sec:bc-bqsm}) can be seen as employing
the repetition code (where one commits to a bit $b$ by measuring in basis
$bb\dots b$).

\begin{quote}
  \rule{\linewidth}{1pt}
  \begin{center}
    {\bf Protocol} \ourbc{}
  \end{center}
  {\bf Setup: }  The generator matrix $G$ of a
  $[N,n,d]$ linear error correcting code.

  {\bf Commit phase:}
  \begin{enumerate}
  \item \verifier{} sends $\ket{x}_\theta$ for $x\in\bool^N$ and
    $\theta\in\{+,\times\}^N$ to the committer.
  \item \committer{} commits to a string $a\in\bool^n$ by measuring each qubit
    $i$ in basis $(G\cdot a)_i$, obtaining a measurement outcome $z\in\bool^N$. 
  \end{enumerate}
  \begin{center}
    \emph{--- Memory bound applies ---}
  \end{center}

  {\bf Reveal phase:} 
  \begin{enumerate}
    \setcounter{enumi}{2}
  \item To open the commitment, \committer{} sends $a$ and $z$ to \verifier{} who
    checks that $z_i=x_i$ whenever $\theta_i=(G\cdot a)_i$.
  \end{enumerate}

  \rule{\linewidth}{1pt}
\end{quote}

Intuitively, we would like the basis $J'$ from Theorem~\ref{thm:all-but-one} to
define the value to which the sender is committed in the sense of
Definition~\ref{def:binding-DFRSS}. The proof would have the verifier purify its
actions and perform the measurement in basis $a$ when the sender opens the
commitment. Theorem~\ref{thm:all-but-one} would ensure the existence of an $a'$
such that the sender is committed to $a'$. There is a subtle issue that prevents
us from applying this argument: the random variable $J'$ whose existence is
stated by Theorem~\ref{thm:all-but-one} exists in the probability space of $X$,
the measurement outcome of the receiver in the opening phase. Therefore, we
cannot assert that $J'$ exists and that the sender is committed to it in the
sense of Definition~\ref{def:binding-DFRSS}. Nevertheless, the techniques
from~\cite{bouman_all-but-one_2013} allows us to prove a weaker statement,
namely that the commitment scheme is sum-binding. 

\begin{theorem}\label{thm:binding}
 The string commitment protocol
  \ourbc{} is sum--biding:
  \begin{equation}
    \sum_a p_a\leq 1+\negl
  \end{equation}
\end{theorem}

\begin{proof}
  We consider an equivalent protocol (from the committer's point of view) where
  the verifier purifies its actions:
  \begin{enumerate}
  \item {\bf Commit phase:} \verifier{} prepares $N$ EPR pairs
    $\bigotimes_{i=1}^N \frac 1{\sqrt 2}(\ket{00}_{P_iV_i}+ \ket{11}_{P_iV_i})$
    and sends registers $P_1\dots P_N$ to \committer{}.
  \item {\bf Reveal phase:} After receiving $(a,z)\in\bool^{2N}$ from
    \committer{}, \verifier{} measures its register $V$ in basis $a$ and checks
    that the result $x$ matches $z$ for each position $i$ in a random sample
    $I\subseteq [N]$.
  \end{enumerate}

  Let $\mathcal{E}_{P\rightarrow EW}$ be the CPTP map describing the partial
  measurement of $\tilde\committer$ after the commit phase, where $\dim E\leq
  2^q$. The joint state of $\verifier$ and $\tilde\committer$ is the density
  operator
  \begin{equation}
    \label{eq:rho-commit}
    \rho_{EWV} := \sum_{w} P_{W}(w) \proj{w} \otimes\rho_{EV}^{w}
    = (\mathcal{E}_P\otimes \id_V)(\ket{EPR}^{\otimes N}_{PV})\enspace.
  \end{equation}
  In general, $\tilde\committer$ may perform a measurement on its quantum
  register $E$ to decide which string $a$ to announce in the reveal phase. The
  most general strategy for $\tilde \committer$ is a POVM
  $\mathcal{M}=\{M^{a,z}_{EW}\}_{(a,z)\in\bool^{2N}}$ where
  $\trace{M^{a,z}\cdot\rho_{EW}}$ gives the probability that $\tilde \committer$
  sends $(a,z)$ in the reveal phase. The probability that $\tilde \committer$
  successfully decommits to $a$ is given by
  \begin{equation}
    \label{eq:pr-accept}
    \Pr[A=a \wedge \verifier \text{ accepts}]= \sum_{z} \trace{M^{a,z}_{EW}\otimes \mathbb{V}_V^{a,z} \rho_{EWV}}
  \end{equation}
where $\mathbb{V}^{a,z}$ is the projective measurement operator corresponding
  to \verifier{}'s check in the reveal phase.

  Consider a fixed $W=w$ and the reduced state $\rho^w_{EV}$. For $a\in\bool^n$,
  let $\mathcal{S}^{a}:= \{x\mid \bra x_a \rho^w_V \ket x_a \leq 2^{-\epsilon
    N}\}$ be the set of outcomes $x$ that have small probability of being
  observed and let $\mathcal{L}^a=\bool^N\setminus \mathcal{S}^a$ its
  complement. Let $Q^a(x)=\bra x_a \rho^w_V \ket x_a$ and
  $Q^a(\mathcal{X})=\sum_{x\in\mathcal{X}}Q^a(x)$ for $\mathcal{X}\subseteq
  \bool^N$. By Theorem 7 of \cite{bouman_all-but-one_2013},
  \begin{equation}
    \label{eq:borne-sum-Q}
    \sum_{a\in\bool^n} Q^a(\mathcal{L}^a) \leq 1+ c\cdot 2^n\cdot \max_{a\neq a'} \sqrt{|\mathcal{L}^a| |\mathcal{L}^{a'}|}
  \end{equation}
  where $c=\max_{a\neq a',x,y}\bra x_a \ket y_{a'}\leq 2^{-\frac d2}$. Since
  $Q^a(x)$ forms a probability distribution over $x$ and $Q^a(x)>2^{-\epsilon
    N}$ for all $x\in\mathcal{L}^a$, we have that $|\mathcal{L}^a|<
  2^{(1-\epsilon)N}$.
  We thus have that \eqref{eq:borne-sum-Q} is bounded above by
  $1+2^{n-d/2+(1-\epsilon)N}$. Let $\eta=2^{n-d/2+(1-\epsilon)N}$.

  Define $\mathbb{L}_a$ and $\mathbb{S}_a$ the projectors onto $\mathcal{L}_a$
  and $\mathcal{S}_a$, respectively. Observe that
  $\mathbb{L}_a+\mathbb{S}_a=\id$. The probability of successful
  opening to \emph{any} $a$ is at most
  \begin{align*}
    &\sum_{z} \trace{M^{a,z}_{E}\otimes \mathbb{V}_V^{a,z} \rho^w_{EV}}\\
    &= \sum_{a,z} \trace{M^{a,z}_{E}\otimes \mathbb{V}_V^{a,z}(\mathbb{L}_a+\mathbb{S}_a) \rho^w_{EV}}\\
    &\leq \sum_{a} \trace{\mathbb{L}_a\cdot \rho^w_{V}}
      + \sum_{a,z} \trace{M^{a,z}_{E}\otimes \mathbb{V}_V^{a,z} \cdot\mathbb{S}_a\cdot \rho^w_{EV}}
  \end{align*}
  The first operand in the sum above corresponds to
  $\sum_aQ^a(\mathcal{L}^a)$ which is bounded above by $1+\eta$. The second
  operand can be upper-bounded by 
  \begin{align*}
    2^q\max_{a,z} \trace{\mathbb{V}_V^{a,z}\cdot\mathbb{S}_a \cdot\rho^w_V}\lesssim 2^{q-\frac{\epsilon}{2} N}
  \end{align*}
  since the trace corresponds to the probability of guessing a random subset of
  a low-probability ($2^{-\epsilon N}$) outcome.

  Assuming $\eta$ can be made $\negl$ with an appropriate choice of parameters,
  if we let $p_a$ denote the probability that $\tilde\committer$ successfully
  opens string $a$, we have that
  \begin{equation}
    \sum_a p_a\leq 1+\negl
  \end{equation}
  \qed
\end{proof}

Observe that to commit to a string of length $n$, protocol \dfss{}
above requires sending $n^2$ qubits from the verifier to the committer.

\section{Witness hiding of $\nizkot[\Pi]$}
\label{sec:witn-hiding}

\begin{definition}[Witness Hiding]\label{def:witn-hiding}Let $R$ be an \npol{} relation, let $G$ be a hard instance generator for $R$
  and let $\Sigma$ be a proof system for $R$. We say that $\Sigma$ is
  \emph{witness hiding} (WH) if there exists a \ppt\ witness extractor $M$ such
  that for any non-uniform \ppt\ $V'$, for any instance $x$,
  \begin{equation*}
    \Pr[(x,w')\in R\mid w'=\langle P(x,w), V'(x) \rangle] \leq \Pr[(x,w')\in R\mid w'=M^{V',G}(x)] +\negl 
  \end{equation*}
  where the probability is (in part) over $x=G(1^n)$.
\end{definition}

We can show that if a $\Sigma$--protocol $\Pi$ is witness hiding, so is
$\nizkot[\Pi]$. We notice that this could also be extended to a
$\Xi$--protocol with an inverse polynomial multiplicative factor on the success
of the extractor $M$.

\begin{theorem}
  If $\Pi$ is a witness hiding $\Sigma$--protocol with $O(\lg \lambda)$--bit
  challenges, then $\nizkot[\Pi]$ is witness hiding\footnote{With a slight
    modification explained in the proof.}.
\end{theorem}
\begin{proof}
  We want to reduce the witness hiding property of $\nizkot[\Pi]$ to that of
  $\Pi$. That is, given a malicious BQS verifier $V$ against $\nizkot[\Pi]$ that
  produces a witness with some probability, we construct a verifier $V_\Pi$
  against $\Pi$ that produces a witness with essentially the same probability.
  For simplicity, we assume challenges are single bits $c\in\bool$. The proof
  for logarithmic length challenges is almost identical.

  Verifier $V_\Pi$ is constructed as follows: in its interaction with the prover
  $P_\Pi$, it selects its challenge $c$ uniformly at random. After the
  interaction with $P_\Pi$, $V_\Pi$ is left with a transcript $(a,c,r)$. Now to
  produce a witness, $V_\Pi$ acts as the prover in an interaction with $V$. It
  prepares and sends the quantum state for the oblivious transfers as $P$ would.
  For its classical message, $V_\Pi$ uses $a$ and $r_c$ from the transcript
  received from $P$ for and sets $r_{1-c}$ to a uniformly random value. By
  Theorem~\ref{thm:bqs-ot}, there exists a random variable $C$ such that the
  value of $r_{1-C}$ is statistically hidden from $V$. With probability
  $\Pr[C=c]$, the view of $V$ in its interaction with $V_\Pi$ will be
  indistinguishable to its view in an interaction with $P$. If $V$ produces a
  valid witness with some probability $p$, the probability that $V_\Pi$ outputs
  $w$ is at least $\Pr[C=c]\cdot p$.

  At this point, an issue occurs if $C$ never takes value $c$, i.e.\
  $\Pr[C=c]=0$ for the particular choice of $c$ by $V_\Pi$. This can easily be
  fixed by having the prover in protocol $\nizkot[\Pi]$ randomize the transcript
  order. With equal probability, the prover uses either $(r_0,r_1)$ or
  $(r_1,r_0)$ as inputs for the OT. The transcript that $V$ receives is now
  uniformly random, such that $\Pr[C=c]=\frac 12$

  In the context of witness hiding, there is no auxiliary input to the verifier,
  so $V_\Pi$ can run $V$ again with the same transcript multiple times such that
  with overwhelming probability, at least one of the runs will provide $V$ with
  the correct view (i.e.\ it will obtain the transcript $(a,c,r)$ that $V_{\Pi}$
  received from $P_{\Pi}$), in which case it will produce a witness with
  probability $p$. The strategy of $V_\Pi$ is to simulate $V$ $k$ times and if
  any of the simulations $V$ produces a witness $w$, $V_\Pi$ outputs $w$. Using
  this strategy, we have
  \begin{align*}
    &\Pr[(x,w')\in R\mid w'\leftarrow \langle P_\Pi(x,w), V_{\Pi}(x) \rangle]\\
    &=\Pr[(x,w')\in R\mid w'\leftarrow \langle V_\Pi(x), V(x) \rangle]\\
    &\geq \Pr[(x,w')\in R\mid w'\leftarrow \langle V_\Pi(x), V(x) \rangle \mid \exists i: C_i=c]\cdot \Pr[\exists i :C_i=c]\\
    &\geq \Pr[(x,w')\in R\mid w'\leftarrow \langle P(x,w), V(x) \rangle] - 2^{-k} - 2^{-\frac n4 +\ell+q} 
  \end{align*}
  where the last inequality follows from the fact that conditioning on $C_i=c$,
  the view of $V$ in the $i$ simulated execution has trace distance at most
  $2^{-\frac n4 +\ell+q}$ from the view in the real execution by
  Theorem~\ref{thm:bqs-ot}.

 \qed
\end{proof}
 
\section{Triviality of Quantum $2$--Message Zero-Knowledge Proofs}
\label{sec:impo-aux}

In this section, we present a quantum version of the impossibility of
zero-knowledge $2$--message \emph{quantum} proof systems for hard languages.
This generalizes the impossibility of~\cite{goldreich_definitions_1994} to
quantum protocols. 
\begin{theorem}
  Let $\Pi=\langle \prover,\verifier \rangle$ be a $2$--message quantum proof
  system for a language $L$. If $\Pi$ is computationally $\epsilon$--sound for
  $\epsilon<\frac 13$ and computationally zero-knowledge, then $L\in \bqp$.
\end{theorem}

  We assume, without loss of generality, that the the verifier is purified,
  i.e., we assume that the general structure of the two-message protocol is as follows:
\begin{enumerate}
\item $\verifier$ prepares a state $\ket\psi_{PV}$ and sends register $P$ to \prover.
\item $\prover$ applies some transform on register $P$ and returns a register $P'$ to $\verifier$.
\item \verifier\ applies a binary-outcome measurement $\{V^x_0,V^x_1\}$ on
  registers $P'V$ and accepts iff outcome is $0$.
\end{enumerate}

Let us assume that this protocol is auxiliary-input quantum ZK, i.e., there
exits a polynomial time quantum simulator $\simulator$ such that  for any
$\tilde\verifier$ the output of $\tilde\verifier$ on input $x$ and $\rho$ in a real
interaction is indistinguishable from
$\simulator_{\tilde \verifier}(x,\rho)$.

Consider the cheating verifier $V^*$ that 
\begin{enumerate}
\item On common input $x$ and auxiliary input register $E$ (of same dimension as
  $P$), sends register $E$ as the first message.
\item On reception of the prover message in quantum register $P'$, output this register $P'$. 
\end{enumerate}
This verifier runs in polynomial time, and so does its simulator. 

Then consider the following \qpt{} machine $M_L$ for deciding if $x\in L$.
Lemmas~\ref{lem:bqp-completeness} and~\ref{lem:bqp-soundness} below show that
this is indeed a \qpt{} algorithm for deciding $L$ which errs with probability
at most $\frac 13$. 
\begin{quote}
  \rule{\linewidth}{1pt}
  \begin{center}
    \bqp{} algorithm $M_L$
  \end{center}

\begin{enumerate}
\item Run the first message function of the honest verifier $\verifier$ on input $x$ to get a register $P$ and an internal register $V$.
\item Run the simulator for $\verifier^*$ on input $x$ and register $P$. Let $P'$ be the output register.
\item Run the verification circuit of $\verifier$ on registers $P'V$. Output ``yes'' if $\verifier$ accepts and ``no'' otherwise.
\end{enumerate}

  \rule{\linewidth}{1pt}
\end{quote}

\begin{lemma}[\bqp{} Completeness]
  \label{lem:bqp-completeness}
  If $\Pi$ is an $\frac 23$--correct quantum auxiliary-input zero-knowledge
  proof of language membership for $L$, then for all $x\in L$, $M_L$ accepts on
  input $x$ with probability at least $\frac 23$.
\end{lemma}
\begin{proof}
  Since $\Pi$ is zero-knowledge, for any cheating verifier $\verifier^*$, there exists a \bqp{} machine $\simulator_{\verifier^*}$ such that the quantum map induced by the interaction of \prover\ and $\verifier^*$ on the auxiliary input of $\verifier^*$ is indistinguishable from the quantum map $\simulator_{\verifier^*}(x,\cdot)$.

  Let $\psi_{PV}=\verifier(x)$ and let $\dist(\rho):=
  \trace{V_0^x \rho}$.
  Let $\Psi_x:=  \prover(x,w)\leftrightharpoons\verifier^*(x,\cdot)$ and $\Phi_x:=
  \simulator_{\verifier^*}(x,\cdot)$ be the real and simulated maps acting on
  the auxiliary information of the verifier. Observe that the quantity $\dist (\Psi_x\otimes
  \id_V(\psi_{PV}))$ corresponds to the probability that the verifier accepts in the real
  protocol and $\dist (\Phi_x \otimes \id_V(\psi_{PV}))$ is the probability
  that $M_L$ accepts on input $x\in L$. By the assumption that the scheme
  is zero-knowledge,
  \begin{equation}
    \|\dist (\Psi^x_P\otimes \id_V(\psi_{PV})) -
    \dist (\Phi^x_P \otimes \id_V(\psi_{PV})) \|
    \leq \negl\enspace .
  \end{equation}
  This means that $M_L$ accepts with essentially the same probability with
  which \verifier{} accepts in the interactive proof, which is at least $\frac 23$.

\qed
\end{proof}

\begin{lemma}[\bqp{} Soundness]
  \label{lem:bqp-soundness}
  If $x\notin L$, then $M_L$ rejects input $x$ with probability $\epsilon>\frac 23$.
\end{lemma}
\begin{proof}
  Consider the cheating prover $\prover^*$ that acts as follows: on common input
  $x$ and register $P$ received from \verifier, compute
  $P'=\simulator_{\verifier^*}(x,P)$ and reply $P'$ to \verifier. Then the
  probability that $\verifier$ accepts in this interaction with a cheating
  prover is equal to the probability that $M_L$ accepts, which by soundness of
  the interactive proof is at most $\epsilon$. \qed
\end{proof}

\end{document}